\documentclass{article}

\usepackage[a4paper, left=3cm, right=3cm, top=3cm]{geometry}
\usepackage{authblk}
\usepackage{amsmath, amsfonts, amssymb, mathtools, amsthm}
\usepackage{bbm}
\usepackage{hyperref}
\usepackage{graphicx}
\usepackage{cite}
\usepackage{enumerate}
\usepackage{enumitem}
\usepackage{multirow}
\usepackage{hhline}
\usepackage{tikz}
\usetikzlibrary{arrows,decorations.pathmorphing,backgrounds,positioning,fit,calc,shapes}

\newcommand{\Uall}{U_{\text{all}}}
\newcommand{\uall}{u_{\text{all}}}
\newcommand{\Vall}{V_{\text{all}}}

\newtheorem{thm}{Theorem}
\newtheorem{lemma}{Lemma}
\newtheorem{cor}{Corollary}

\title{Confirmatory adaptive group sequential designs for clinical trials with multiple time-to-event outcomes in Markov models}

\author[1,*]{Moritz Fabian Danzer}
\author[1]{Andreas Faldum}
\author[2]{Thorsten Simon}
\author[2]{Barbara Hero}
\author[1]{Rene Schmidt}

\affil[1]{Institute of Biostatistics and Clinical Research, University of Münster, Münster, Germany}
\affil[2]{Department of Pediatric Oncology and Hematology, University Hospital Cologne, Cologne, Germany}
\affil[*]{Corresponding author, moritzfabian.danzer@ukmuenster.de}

\date{\today}

\begin{document}
	
\maketitle

\begin{abstract}
The analysis of multiple time-to-event outcomes in a randomised controlled clinical trial can be accomplished with exisiting methods. However, depending on the characteristics of the disease under investigation and the circumstances in which the study is planned, it may be of interest to conduct interim analyses and adapt the study design if necessary. Due to the expected dependency of the endpoints, the full available information on the involved endpoints may not be used for this purpose. We suggest a solution to this problem by embedding the endpoints in a multi-state model. If this model is Markovian, it is possible to take the disease history of the patients into account and allow for data-dependent design adaptiations.\\
To this end, we introduce a flexible test procedure for a variety of applications, but are particularly concerned with the simultaneous consideration of progression-free survival (PFS) and overall survival (OS). This setting is of key interest in oncological trials. We conduct simulation studies to determine the properties for small sample sizes and demonstrate an application based on data from the NB2004-HR study.
\end{abstract}

\textbf{Keywords:} clinical trial, log-rank test, sample size recalculation, survival analysis

\section{Introduction}
Adaptive clinical trial designs for a single primary time-to-event endpoint are well-established (see e.g. \cite{Muller:01, Wassmer:06}). These are based on the log-rank test by exploiting its independent increments structure as exhibited in \cite{Tsiatis:1981, Sellke82}. As long as only information on this single endpoint is used to inform an adaptation of the design in an interim analysis, the nominal type 1 error rate will be maintained. However, this no longer applies if information of further endpoints is used from patients, who have been recruited before this interim analysis and remain event-free beyond it \cite{BPB:02}. This is because these additional data can be used to predict the course of the disease in those same patients. For example, information on progression status can be used to predict individual mortality risk in a trial with primary endpoint overall survival (OS). Such misuse of surrogate interim data leads to inflation of the actual type I error level. Approaches to solving this problem make use of the strategy of patient-wise separation \cite{Jenkins:11, Irle:12, Jörgens:2019}. Although the initial approaches in \cite{Jenkins:11, Irle:12} have already been improved by \cite{Jörgens:2019}, some disadvantages cannot be resolved, such as partial discarding of primary endpoint data in the final analysis. Alternatively, worst-case adjustments can be made to avoid a type I error inflation \cite{Magirr:16} that often result in a conservative procedure.\\
Similar issues arise as well for trials with multiple primary time-to-event endpoints. For one-sample studies, this situation has already been addressed by \cite{Danzer:2022}. Single-stage procedures for the simultaneous assessment of multiple time-to-event endpoints in randomised trials have already been proposed in \cite{Wei:1984}. Roughly speaking, this method can be described as performing separate log-rank tests simultaneously for all endpoints involved. A final test decision is made by looking at the joint distribution of the individual test statistics. Corresponding group sequential procedures were introduced in \cite{Lin:1991}. At first glance, an extension of \cite{Lin:1991} to adaptive designs seems obvious by following the strategy of \cite{Wassmer:06}, since the multivariate test statistic also has a property of independent increments. However, this property only applies to each component of the multivariate test statistic separately and not for the multivariate process as a whole. The reason for this is closely linked to the problem mentioned in \cite{BPB:02} because, again, information about some endpoints might be used to predict future outcome of other endpoints. At the same time, patients who are known to be in different disease states are compared to each other. We will solve this problem by taking into account the available information on all endpoints when calculating the test statistics, thus only comparing patients who have the same prognosis of disease course given the available information.\\
To this end, it is central to our approach that we can easily embed different time-to-event endpoints into a multi-state model. Especially in oncology, which is of central importance to us as an area of application for our methods, such models can be very helpful in being able to depict different courses of disease \cite{Le-Rademacher:2018}. Two of the most important endpoints in this field of clinical research are given by progression-free survival (PFS)/event-free survival (EFS) and overall survival (OS). While the latter one is the most objectively defined endpoint, the former can often be regarded as its surrogate and has certain advantages in terms of time- and cost-effectiveness. The exact definition of the endpoint and its suitability as a primary endpoint strongly depends on the tumour entity and the patient collective to be considered \cite{Bellera:2013}. Those two endpoints can be embedded in a simple illness-death model which has been discussed extensively in \cite{Meller:2019}. Provided that this model has the Markov property, we can perform a two-group comparison that addresses the aforementioned issues. As in \cite{Lin:1991}, this results in a consideration that refers to the clinical endpoints with the aid of a transitional consideration as in \cite{Tattar:2014}.\\
The paper is organised as follows. It starts with a presentation of the procedure for the prominent example of PFS and OS. Sections \ref{sec:notation} and \ref{sec:general_framework} introduce the general notation and generalize the procedure for broader applications. Building on that, planning and execution of a clinical trial is briefly sketched in Section \ref{sec:gs_design}. Properties of the method in practically relevant scenarios are studied by simulation in Section \ref{sec:simulation_study}. An application of the proposed method is demonstrated in Section \ref{sec:application_example} using the data from the NB2004-HR trial (NCT number NCT03042429). We conclude with a discussion of our findings and prospects for future research.\\
Proofs of mathematical statements and a further case study are shifted to the Appendix.

\section{Main application example: PFS \& OS}\label{sec:pfs_os}
We illustrate our procedure using the example of a trial with the primary time-to-event endpoints PFS and OS. In a randomised clinical trial, PFS is defined as time from randomisation to progression of the disease or death, whatever occurs first. OS denotes the time from randomisation to death. While OS is obviously the most objectively defined time-to-event endpoint, the use of other endpoints such as PFS may also be justified in oncological phase III clinical trials, depending on the nature of the disease and the mechanism of action of the experimental treatment. Outcome improvement can first be associated with longer progression-free survival time or an increase of the rate of patients without tumor progression. In addition, there may be other advantages concerning death without prior progression or post-progression survival which then additionally affect OS. The methods proposed here cover all of these aspects by allowing to use both of these endpoints as primary endpoints under exploitation of their dependence structure in a Markovian multi-state model.\\
Such a model as presented extensively in \cite{Meller:2019} establishes the corresponding probabilistic structure. The multi-state model is visualized in Figure \ref{fig:pfsos}. A patient's  history of disease from start of the therapy corresponds to a path along the arrows in this figure. At the beginning of the treatment, a patient starts in state $0$. He may die directly without progression. This is represented by a jump to state $2$. Otherwise, he may experience a progression of the disease, which is represented by a jump to state $1$ and die afterwards which is represented to a subsequent jump to state $2$.

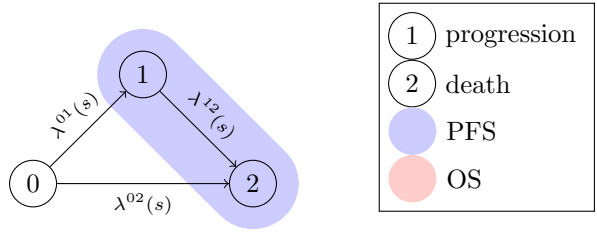
\begin{figure}[h]
	\centering
	\begin{tikzpicture}
		\node[circle, draw] (initial) {$0$};
		\node[circle, draw] (progression) [above right=of initial] {$1$};
		\node[circle, draw] (death) [below right=of progression] {$2$};
		
		\draw [->] (initial) -- node[midway, sloped, above] {\scriptsize $\lambda^{01}(s)$} (progression);
		\draw [->] (initial) -- node[midway, below] {\scriptsize $\lambda^{02}(s)$} (death);
		\draw [->] (progression) -- node[midway, sloped, above] {\scriptsize $\lambda^{12}(s)$} (death);
		
		\begin{pgfonlayer}{background}
			\foreach \nodename in {progression, death} {
				\coordinate (\nodename') at (\nodename);
			}
			\path[fill=blue!20, draw=blue!20, line width = 1.2cm, line cap=round, line join=round] 
			(progression') to[bend left=0] (death') 
			to[bend left=0] (progression') -- cycle;
			\path[fill=red!20, draw=red!20, line width = 1.0cm, line cap=round, line join=round] 
			(death') to[bend left=0] (death') -- cycle;
		\end{pgfonlayer}
		
		\matrix [draw,below right,xshift = 30pt, yshift = 10pt] at (current bounding box.north east) {
			\node [circle, draw, label=right:progression] {$1$}; \\
			\node [circle, draw, label=right:death] {$2$}; \\
			\node [circle, fill=blue!20, label=right:PFS] {$\quad$};\\
			\node [circle, fill=red!20, label=right:OS] {$\quad$};\\
		};
		
	\end{tikzpicture}
	\caption{Representation of the PFS/OS scenario as a multi-state model}
	\label{fig:pfsos}
\end{figure}
In accordance with our general framework, we denote the random time of transition to node $1$ of some patient $i$ by $T_i^{\{1\}}$ and the time of transition to node $2$ by $T_i^{\{2\}}$. Accordingly, the random time of PFS which is the first hitting time of the set of nodes $\{1,2\}$ can be defined as $T_i^{\text{PFS}}\coloneqq T_i^{\{1\}} \wedge T_i^{\{2\}}$ where $a\wedge b$ denotes the minimum of two real numbers $a,b$. The random time of OS which is the first hitting time of node 2 is given by $T_i^{\text{OS}}\coloneqq T_i^{\{2\}}$.\\
Such a model fulfills the Markov assumption if the conditional probability of future transitions does only depend on the present state. To introduce this more formally, let $X_i\colon \mathbb{R}_+ \to \{0,1,2\}$ denote the state occupation function for some patient $i$, i.e. $X_i(s)$ yields the state of patient $i$ at time $s$ since randomisation of that patient. The left-hand limits $X_i(s-)\coloneqq \lim_{h \searrow 0} X_i(s-h)$ denote the state of the patient just before $s$. Now, if
\begin{equation*}
	\mathbb{P}\left[X_i(s_2) \in S|(X_i(u))_{u \in [0,s_1]}\right] = \mathbb{P}[X_i(s_2) \in S|X(s_1)]
\end{equation*}
for any subspace of the state space $S\subset \{0,1,2\}$ and any $0\leq s_1 < s_2$, the stochastic process $(X_i(s))_{s \geq 0}$ is said to possess the Markov property.\\
Given the current state of a patient, its instantaneous rate of transition to another state does only depend on the time elapsed since randomisation. Hence, each of the transitions represented by the arrows is equipped with a univariate transition hazard or intensity function $\lambda^{jk}\colon \mathbb{R}_+ \to \mathbb{R}_+$ with $0 \leq j < k \leq 2$. These are defined by
\begin{equation*}
	\lambda^{jk}(s)\coloneqq \lim_{h \searrow 0} \frac{\mathbb{P}[X_i(s+h)=k | X_i(s) = j]}{h}.
\end{equation*}
Given those functions, the joint distribution of PFS and OS is fully specified. As we want to explore any kind of possible differences in the joint distribution of the two endpoints between the two groups, our testing procedure investigates the null hypothesis

\begin{equation}\label{eq:H0_dist_fct_pfsos}
	H_0\colon F^0_{\text{PFS,OS}} = F^1_{\text{PFS,OS}}.
\end{equation}

where $F^g_{\text{PFS,OS}}$ denotes the joint distribution function of the time-to-event endpoints PFS and OS in group $g$. Following \cite{Tattar:2014}, this could also be reformulated in terms of the cumulative intensity matrices $\mathbf{\Lambda}^g$. The $(j,k)$-th entry of this matrix is the function of the cumulative intensity for the transition from state $j$ to $k$ for the respective group $g$. The corresponding hypothesis is 

\begin{equation}\label{eq:H0_cum_int_pfsos}
	H_0\colon \mathbf{\Lambda}^0 = \mathbf{\Lambda}^1.
\end{equation}

Differing from \cite{Tattar:2014}, we do not compare the estimated transition intensity matrices, but pursue an approach that is motivated by the clinically relevant endpoints.\\
In univariate survival analysis, one-dimensional compensated counting processes form the basis for constructing adaptive designs. For the two endpoints considered here, these are given by $(\tilde{M}_i^{\text{PFS}}(s))_{s\geq 0}$ respectively (resp.) $(\tilde{M}_i^{\text{OS}}(s))_{s\geq 0}$ with
\begin{equation}\label{eq:univariate_martingale}
	\tilde{M}_i^E(s)\coloneqq \mathbbm{1}_{\{T^E_i\leq s \wedge C_i\}} - \int_0^{s\wedge T_i^E \wedge C_i} \lambda^E (u)\, du 
\end{equation}
for $E \in \{\text{PFS}, \text{OS}\}$ and any $s\geq 0$. The positive real-valued random variable $C_i$ denotes the random censoring time which is assumed to be independent from the process $X_i$. The endpoint-specific hazards $\lambda^{\text{PFS}}$ and $\lambda^{\text{OS}}$ do not take into account the current state of the patient. Since we have to do exactly this when constructing adaptive designs where all information on PFS and OS is allowed to be used at an interim analysis, we will instead consider the multivariate compensated counting processes $(\mathbf{M}_i(s))_{s\geq 0}$ with
\begin{equation*}
	\mathbf{M}_i(s) \coloneqq 
	\begin{pmatrix}
		M_i^{\text{PFS}}(s)\\
		M_i^{\text{OS}}(s)	
	\end{pmatrix}
	=
	\begin{pmatrix}
		\tilde{M}_i^{\text{PFS}}(s)\\
		M_i^{\text{OS}}(s)	
	\end{pmatrix}
\end{equation*}
and 
\begin{equation}\label{eq:os_martingale_mv}
	M_i^{\text{OS}}(s)\coloneqq \mathbbm{1}_{\{T^{\text{OS}}_i\leq s \wedge C_i\}} - \int_0^{s\wedge C_i \wedge T_i^{\text{PFS}}} \lambda^{02}(u)\, du - \int_{s\wedge C_i \wedge T_i^{\text{PFS}}}^{s\wedge C_i \wedge T_i^{\text{OS}}} \lambda^{12}(u)\, du
\end{equation}
for any $s\geq 0$. The component for PFS can be adopted from the univariate setting (according to \eqref{eq:univariate_martingale}) as there is no additional information to be taken into account for this endpoint. As soon as any transition occurs in our simple model, the process automatically stops.\\
In order to state the test statistics that arise in this way, we need to introduce some more notation. First, let $Z_i\in \{0,1\}$ denote the treatment indicator variable and $R_i \in \mathbb{R}_+$ the random time of trial entry of patient $i$. As we aim for adaptive sequential designs, we need to deal with two different time scales: We will always denote the calendar time by $t$ and the individual time in trial by $s$. In this way, we can define the event counting processes
\begin{equation*}
N^{\text{PFS}}_i(t,s) \coloneqq \mathbbm{1}_{\{T_i^{\text{PFS}} \leq s \wedge C_i \wedge (t-R_i)_+\}} \quad \text{and} \quad N^{\text{OS}}_i(t,s) \coloneqq \mathbbm{1}_{\{T_i^{\text{OS}} \leq s \wedge C_i \wedge (t-R_i)_+\}}
\end{equation*}
counting events that happen before calendar time $t$ and trial time $s$. For any state $l\in \{0,1,2\}$ of our model from Figure \ref{fig:pfsos} we can also define the corresponding at risk processes
\begin{equation*}
Y_i^{j}(t,s)\coloneqq \mathbbm{1}_{\{X_i(s-)=j\}} \cdot \mathbbm{1}_{\{s \leq C_i \wedge (t-R_i)_+\}} \quad \text{and} \quad Y_i^{j,Z=1}(t,s)\coloneqq Z_i \cdot Y_i^{j}(t,s)
\end{equation*}
which indicate at some calendar time $t$ whether patient $i$ is known to be in state $l$ just before trial time $s$ and (for the latter one) whether the patient is in treatment group 1. While these quantities are defined for each patient, the aggregates $N^{\text{PFS}},\, N^{\text{OS}},\,Y^{j}$ and $Y^{j, Z=1}$ over the whole study sample are given by summing the corresponding quantities over all patients $i$ from 1 to $n$.\\
At calendar time $t$ the component of our unstandardised multivariate test statistic concerning PFS is then given by
\begin{equation*}
	U^{\text{PFS}}(t)\coloneqq \frac{1}{\sqrt{n}} \sum_{i=1}^n \int_0^t \left(Z_i - \frac{Y^{0,Z=1}(t,s)}{Y^{0}(t,s)}\right) N_i^{\text{PFS}}(t,ds).
\end{equation*}
which is just the common unstandardised log-rank statistic for PFS. For the second component, concerning the endpoint OS, we need to take the additional information of prior progressions into account. It is defined by
\begin{equation*}
	U^{\text{OS}}(t)
	\coloneqq \frac{1}{\sqrt{n}}\sum_{i=1}^n \int_0^t \left(Z_i - Y_i^0(t,s) \frac{Y^{0,Z=1}(t,s)}{Y^{0}(t,s)} - Y_i^1(t,s) \frac{Y^{1,Z=1}(t,s)}{Y^{1}(t,s)}\right) N_i^{\text{OS}}(t,ds).
\end{equation*}
Analogously to the adopted compensated counting process in \eqref{eq:os_martingale_mv}, we need to distinguish between patients who did not experience a progression of the disease yet $(Y_i^0(t,s)=1)$ and those who did $(Y_i^1(t,s)=1)$. In contrast to \cite{Lin:1991}, this distinction enables adaptive design changes based on all information from the illness-death model from Figure \ref{fig:pfsos}.\\
The variance of $U^{\text{PFS}}(t)$ and $U^{\text{OS}}(t)$ can be estimated by
\begin{equation*}
	\hat{V}^{\text{PFS}}(t)=\frac{1}{n} \sum_{i=1}^n \int_{[0,t]} \frac{Y^{0,Z=1}(t,s)}{Y^{0}(t,s)} \left(1 - \frac{Y^{0,Z=1}(t,s)}{Y^{0}(t,s)} \right) N_i^{\text{PFS}}(t,ds)
\end{equation*}
resp.
\begin{align*}
	\hat{V}^{\text{OS}}(t)=\frac{1}{n} \sum_{i=1}^n \int_{[0,t]} &Y_i^0(t,s) \frac{Y^{0,Z=1}(t,s)}{Y^{0}(t,s)}\left(1 - \frac{Y^{0,Z=1}(t,s)}{Y^{0}(t,s)}\right)\\
	& + Y_i^1(t,s) \frac{Y^{1,Z=1}(t,s)}{Y^{1}(t,s)} \left(1 - \frac{Y^{1,Z=1}(t,s)}{Y^{1}(t,s)}\right) N_i^{\text{OS}}(t,ds)
\end{align*}
The covariance between the two random variables can be estimated by
\begin{align*}
	\hat{V}^{\text{PFS,OS}}(t)=\frac{1}{n} \sum_{i=1}^n \int_{[0,t]} Y_i^0(t,s) \frac{Y^{0,Z=1}(t,s)}{Y^{0}(t,s)}\left(1 - \frac{Y^{0,Z=1}(t,s)}{Y^{0}(t,s)}\right) N_i^{\text{OS}}(t,ds).
\end{align*}
As we integrate over the process $N_i^{\text{OS}}(t,ds)$ and multiply each summand by $\mathbbm{1}_{\{Y_i^0(t,s-)=1\}}$, we can see that the covariance is only driven by transitions from state 0 to node 2, i.e. by patients for which PFS and OS happen simultaneously.

We will now consider the bivariate process $(\mathbf{U}(t))_{t \geq 0}$ with $\mathbf{U} (t) \coloneqq (U^{\text{PFS}}(t), U^{\text{OS}}(t))$ for all $t \geq 0$ and the $2\times 2$-matrix valued process $(\hat{\mathbf{V}} (t))_{t \geq 0}$ with
\begin{equation*}
	\hat{\mathbf{V}} (t) = 
	\begin{pmatrix}
		\hat{V}^{\text{PFS}}(t) & \hat{V}^{\text{PFS,OS}}(t)\\
		\hat{V}^{\text{PFS,OS}}(t) & \hat{V}^{\text{OS}}(t)
	\end{pmatrix}
\end{equation*}
for all $t \geq 0$.

For the sake of simplicity, we only consider a design with one interim analysis at calendar time $t_1>0$ and a final analysis at calendar time $t_2 > t_1$ here. First stage test statistics will be based on $\mathbf{U} (t_1)$ and $\hat{\mathbf{V}} (t_1)$. Test statistics for the data from the second stage will be based on the increments since calendar time $t_1$, i.e. $\mathbf{U} (t_2) - \mathbf{U} (t_1)$ and $\hat{\mathbf{V}} (t_2) - \hat{\mathbf{V}} (t_1)$. If $\hat{\mathbf{L}}(t_i, t_{i-1})$ denotes the lower triangular Cholesky factor of $\hat{\mathbf{V}}(t_i) - \hat{\mathbf{V}}(t_{i-1})$ where $t_0\coloneqq 0$, then the two multivariate test statistics
\begin{equation*}
	\mathbf{Z}_1 \coloneqq \hat{\mathbf{L}}(t_1, t_0)^{-1} \mathbf{U} (t_1) \quad \text{and} \quad \mathbf{Z}_2 \coloneqq \hat{\mathbf{L}}(t_2, t_1)^{-1} (\mathbf{U} (t_2) - \mathbf{U} (t_1))
\end{equation*}
are asymptotically bivariate standard normally distributed and independent. To obtain a univariate test statistic for each stage one can use the $L^2$-norm $|z|_2\coloneqq \sqrt{z_1^2 + z_2^2}$. This yields the stagewise $p$-values
\begin{equation}\label{eq:stagewise_p_values_pfsos}
	p_r = 1 - F_{\chi_2^2}(|Z_r|_2^2)
\end{equation}
for $r \in \{1,2\}$. Thus we obtain $p$-values in analogy to \cite{Wei:1984} as we can rewrite
\begin{equation*}
	|Z_1|_2^2 = \mathbf{U}(t_1)^T \hat{\mathbf{V}}(t_1)^{-1}\mathbf{U}(t_1)\quad \text{resp.} \quad |Z_2|_2^2 = (\mathbf{U}(t_2)-\mathbf{U}(t_1))^T (\hat{\mathbf{V}}(t_2) - \hat{\mathbf{V}}(t_1))^{-1}(\mathbf{U}(t_2)-\mathbf{U}(t_1)).
\end{equation*}
The stagewise $p$-values can then be further processed using the standard methods for adaptive designs of clinical trials.

\section{General framework}\label{sec:notation}
In this section we will introduce the framework and all its components we need to construct the multivariate process and resulting test statistics. This will allow us to expand upon the example from the previous section by considering an arbitrary number of composite events.\\
Let $(\Omega, \mathcal{F},\mathbb{P})$ denote the probability space upon which all random variables are defined. Any patient $i$ enters the trial at the random time $R_i\geq 0$ and is assigned to treatment group $Z_i \in \{0,1\}$. During the stay in the trial, the patients assume one state of the state space $\{0,1,\dots,l\}$. The assumed state may change in course of time. For each $i$ and any $s \geq 0$, let $X_i(s) \in \{0,1,\dots,l\}$ denote the individual state of patient $i$ at time $s$ after its trial entry. The multistate process $(X_i(s))_{s \geq 0}$ is assumed to be càdlàg. Each individual starts in the initial node $0$, i.e. $X_i(0)=0$ for any $i$. We assume that $(X_i(s))_{s \geq 0}$ is Markovian, i.e.
\begin{equation}\label{eq:markov_property}
	\mathbb{P}\left[X_i(s_2) \in S|(X_i(u))_{u \in [0,s_1]}\right] = \mathbb{P}[X_i(s_2) \in S|X(s_1)]
\end{equation}
for any $S\subset \{0,1,\dots,l\}$ and any $0\leq s_1 < s_2$. This means that the probability of any future transition depends on the past only via the current time $s_1$ and the current state $X(s_1)$. Given the Markov property, the probabilities of transitions of the patient are completely determined by the transition intensities $\lambda^{jk} \colon \mathbb{R}_+ \to \mathbb{R}_+$ for $j,k \in \{0,\dots,l\}$ defined by
\begin{equation*}
	\lambda^{jk}(s)\coloneqq \lim_{h \searrow 0} \frac{\mathbb{P}[X_i(s+h) = k|X_i(s-) = j]}{h}
\end{equation*}
where $X_i(s-)$ again denotes the left hand limit of the càdlàg process. When illustrating such a model as in Figure \ref{fig:pfsos}, we only draw arrows from a state $k$ to some other state $j$ if $\lambda^{jk} \not\equiv 0$. 
As there may be some terminal or absorbing state $j$ in the model (for example the state of death in the model of Section \ref{sec:pfs_os}), these terminal states have the property that
\begin{equation*}
	\lambda^{jk}\equiv 0 \quad \forall k\neq j.
\end{equation*}
On that basis, we can define hitting or first entry times for each node $j \in \{1,\dots,l\}$ by
\begin{equation*}
	T_i^{\{j\}}\coloneqq \inf\{s \geq 0 | X_i(s) = j\}.
\end{equation*}
If the time of a composite event is of interest, this can be depicted by the hitting time of a set of nodes $E \subset \{1,\dots,l\}$ with
\begin{equation*}
	T_i^{E}\coloneqq \inf\{s \geq 0 | X_i(s) \in E\}.
\end{equation*}
However, the observations for all our patients can be censored. Either by administrative censoring at the time of an interim or the final analysis or by random dropout. In the former case, an analysis at calendar time $t$ induces an adminsitrative censoring at $(t-R_i)_+$. The latter case is depicted by the random variable $\tilde{C}_i$. Combining this information at calendar time $t$ yields the censoring variable $C_i(t)\coloneqq \tilde{C}_i \wedge (t-R_i)_+$. Note that censoring by some terminal event as e.g. death is not included here.\\ 
We assume the tuples $(R_i, Z_i, \tilde{C}_i, (X_i(s))_{s \geq 0})$ for $i \in \{1,\dots,n\}$ to be independet replicates of some tuple $(R, Z, \tilde{C}, (X(s))_{s \geq 0})$. Additionally, we assume independent censoring and recruitment mechanisms, i.e. that the variable $Z$, $R$ and $\tilde{C}$ are mutually independent.\\
With the quantities given above, we can now define counting processes and at risk indicators for the occurence of certain events. First, for any event given via a set $E \subset \{0,\dots,k\}$, the multivariable process $(N^{E}_i(t,s))_{t\geq 0, s\geq 0}$ defined by
\begin{equation*}
	N^{E}_i(t,s) \coloneqq \mathbbm{1}_{\{T_i^{E} \leq s \wedge C_i(t)\}}
\end{equation*}
indicates whether a visit of patient $i$ in the subset $E$ of the state space (resp. the event associated with this set) has been observed before calendar time $t$ and trial time $s$. We can also aggregate these individual counting processes to obtain the overall number of events $N^E(t,s)\coloneqq \sum_{i=1}^n N^{E}_i(t,s)$ observed before calendar time $t$ and trial time $s$.\\
As indicated by the Markov property in \eqref{eq:os_martingale_mv}, the current state of a process at some trial time $s$ determines the probability of future transitions. Hence, it will be of utmost importance for our procedure to keep track of the current state of each individual. The multivariable process $(Y_i^j(t,s))_{t\geq 0, s\geq 0}$ indicates whether it is known at calendar time $t$ that individual $i$ has been in state $j$ just before its trial time $s$. It is defined by
\begin{equation*}
	Y_i^j(t,s)\coloneqq \mathbbm{1}_{\{X_i(s-) = j\}} \cdot \mathbbm{1}_{\{s \leq C_i(t)\}}
\end{equation*} 
We can aggregate these indicators in the complete study sample or in the subsample of treatment group 1 to obtain the processes $(Y^j(t,s))_{t\geq 0, s\geq 0}$ resp. $(Y^{j, Z=1}(t,s))_{t\geq 0, s\geq 0}$ counting the number of patients in state $j$ with
\begin{equation}\label{eq:at_risk_process}
	Y^j(t,s)\coloneqq \sum_{i=1}^n Y_i^j(t,s) \quad \text{resp.} \quad Y^{j,Z=1}(t,s)\coloneqq \sum_{i=1}^n Z_i \cdot Y_i^j(t,s).
\end{equation}
As we only consider the first hitting time of a subset $E$ of the state space which corresponds to the event time of the corresponding (composite) event for now, we need to restrict these at risk numbers to those patients which did not already experience the event $E$. Those quantities are given by $(Y^{j \to E}_i(t,s))_{t\geq 0, s\geq 0}$ resp. $(Y^{j \to E, Z=1}_i(t,s))_{t\geq 0, s\geq 0}$ for any patient $i$. Those quantities are defined by
\begin{equation*}
	Y_i^{j \to E}(t,s)\coloneqq Y_i^j(t,s)\cdot \mathbbm{1}_{\{T_i^E \geq s\}} \quad \text{resp.} \quad Y_i^{j\to E,Z=1}(t,s)\coloneqq Y_{i,Z=1}^j(t,s)\cdot \mathbbm{1}_{\{T_i^E \geq s\}}.
\end{equation*}
and the aggregates $(Y^{j \to E}(t,s))_{t\geq 0, s\geq 0}$ resp. $(Y^{j \to E, Z=1}(t,s))_{t\geq 0, s\geq 0}$ over the whole study sample are as in \eqref{eq:at_risk_process}.

\section{Construction of the multivariate process and its asymptotics}\label{sec:general_framework}
First, we consider only one composite event represented by a subspace of the state space except the initial state $E \subset \{1,\dots,k\}$. For this event, we define the stochastic process
\begin{equation*}
	U^E(t)\coloneqq \frac{1}{\sqrt{n}} \sum_{i=1}^n \int_{[0,t \wedge T_i^E]} \left(Z_i -\sum_{j \notin E} Y_i^{j \to E}(t,u) \frac{Y^{j \to E, Z=1}(t,u)}{Y^{j \to E}(t,u)} \right) N_i^E(t,du).
\end{equation*} 
Different from a standard log-rank test for the composite endpoint $E$ we need to distinguish between the states from which a transition to one of the component events belonging to $E$ occurs.\\
Now, let $d \in \mathbb{N}$ composite events of interest be given via some subsets of the state space $E_1,\dots,E_d$. Similar to the formulation of the null hypothesis for the example in Section \ref{sec:pfs_os}, we can now formulate the null hypothesis in terms of the joint distribution by
\begin{equation}\label{eq:H0_dist_fct}
	H_0\colon F^0_{T^{E_1}, \dots, T^{E_d}} = F^1_{T^{E_1}, \dots, T^{E_d}}
\end{equation}
or in terms of the cumulative transition intensity matrix by
\begin{equation}\label{eq:H0_cum_int}
	H_0\colon \mathbf{\Lambda}^0 = \mathbf{\Lambda}^1.
\end{equation}
To test these hypotheses, we will consider the multivariate process $\mathbf{U}\colon \mathbb{R}_+ \to \mathbb{R}^d$
\begin{equation*}
	\mathbf{U}(t) = (U^{E_1}(t),\dots,U^{E_d}(t)).
\end{equation*}
In Corollary \ref{cor:asymp_equiv_comp} it is shown that this process is asymptotically equivalent to a martingale with respect to the filtration incorporating any information about events in the multi-state model. Please note that it is the same as the multivariate process introduced in \cite{Lin:1991} in a competing risks setting, where each of the states $1,\dots,k$ is a terminal state. In this special case, there is no difference between the two methods as there are no intermediate events which can be used to make predictions about later events of the same patient.
The variance of $U^E(t)$ can then be estimated by
\begin{equation*}
	\hat{V}^E(t) \coloneqq \frac{1}{n} \sum_{i=1}^n \int_{[0,t \wedge T_i^E]} \sum_{j \notin E} \left( Y_i^{j \to E}(t,u) \cdot \frac{Y^{j \to E, Z=1}(t,u)}{Y^{j\to E}(t,u)} \left(1 - \frac{Y^{j\to E, Z=1}(t,u)}{Y^{j\to E}(t,u)}\right) \right) N_i^E(t,du)
\end{equation*}
If at least two of the sets $E_1,\dots,E_d$ have a non-empty intersection, i.e. if two of the composite events may occur at the same time, there is a non-zero covariance of the corresponding entries in $\mathbf{U}(t)$. Accordingly for $b, c \in \{1,\dots,d\}$, the covariance $\text{Cov}(U^{E_b}(t), U^{E_c}(t))$ can be estimated by
\begin{align*}
	&\hat{V}^{E_b E_c}(t)\\
	\coloneqq &\frac{1}{n} \sum_{i=1}^n \int_{[0,t \wedge T_i^{E_b \cup E_c}]} \sum_{j \notin E} \left( Y_i^{j \to E_b \cup E_c}(t,u)\cdot \right. \\
	&\qquad \qquad \qquad \left. \frac{Y^{j \to E_b \cup E_c, Z=1}(t,u)}{Y^{j \to E_b \cup E_c}(t,u)} \left(1 - \frac{Y^{j \to E_b \cup E_c, Z=1}(t,u)}{Y^{j \to E_b \cup E_c}(t,u)}\right) \right) N_i^{E_b \cap E_c}(t,du).
\end{align*}
The covariance of the process $\mathbf{U}$ is thus estimated by the $d\times d$-matrix valued function $\hat{\mathbf{V}}\colon \mathbb{R}_+ \to \mathbb{R}_+^{d\times d}$ with
\begin{equation*}
	\hat{\mathbf{V}}(t)= \left(\hat{V}^{E_b E_c}(t)\right)_{1 \leq b,c \leq d}.
\end{equation*}
Invertibility of increments of the variance matrix $\mathbf{V}$ given in Theorem \ref{thm:rebolledo_application} (which are estimated by the corresponding increments of $\hat{\mathbf{V}}$) should be checked first. Invertibility is of course not given if e.g. $E_b = E_c$ for some $b\neq c$ and $b,c \in \{1,\dots,d\}$. However, in the Appendix we provide criteria for invertibility of $\mathbf{V}$. In most relevant cases as e.g. in those mentioned in the main manuscript, these criteria can easily be verified to be fulfilled.\\
In a group sequential design, there is a sequence of calendar dates $t_1,\dots,t_m$ at which analyses shall be conducted. It is a result of Theorem \ref{thm:transition_stat} that the increments of the process $\mathbf{U}$ are asymptotically independent and the covariance matrix of these increments can consistently be estimated by increments of $\hat{\mathbf{V}}$. Now, we consider the lower triangular Cholesky factor $\hat{\mathbf{L}}_r$ of the increment $\hat{\mathbf{V}}(t_r) - \hat{\mathbf{V}}(t_{r - 1})$ for $r \in \{1,\dots,m\}$ with $t_0=0$. If we define 
\begin{equation*}
	\mathbf{Z}_r \coloneqq \hat{\mathbf{L}}_r (\mathbf{U}(t_r) - \mathbf{U}(t_{r - 1})),
\end{equation*}
we obtain the asymptotical result
\begin{equation*}
	\mathbf{Z}\coloneqq(\mathbf{Z}_1, \dots, \mathbf{Z}_m) \overset{\mathcal{D}}{\to} \mathcal{N}(0, \mathbf{1}_{d\cdot m})
\end{equation*} 
as $n \to \infty$ where $\mathbf{1}_{d\cdot m}$ denotes the identity matrix of size $d\cdot m$.\\
Based on this, various stagewise test statistics can be constructed. With reference to \cite{Wei:1984}, we propose 
\begin{equation*}
	S_r \coloneqq (\mathbf{U}(t_{r}) - \mathbf{U}(t_{r-1}))^T (\hat{\mathbf{V}}(t_{r}) - \hat{\mathbf{V}}(t_{r-1}))^+ (\mathbf{U}(t_{r}) - \mathbf{U}(t_{r-1})) {\overset{\mathcal{D}}{\to}} \chi^2_{d} \qquad \forall r \in \{1,\dots,m\}
\end{equation*}
as a natural test statistic for testing $H_0$ in stage $r\in \{1,\dots,m\}$. Here, $\mathbf{A}^+$ denotes the Moore-Penrose inverse of a quadratic matrix $\mathbf{A}$. Following Corollary \ref{cor:test_stat}, $S_1,\dots,S_m$ are asymptotically independent and asymptotically follow a $\chi^2$-distribution with $d$ degrees of freedom. Stagewise $p$-values can accordigly be computed by
\begin{equation}\label{eq:stagewise_p_value}
	p_r \coloneqq 1 - F_{\chi^2_d}(S_r)
\end{equation}
for any $r \in \{1,\dots,m\}$.\\
Going beyond the joint assessment of PFS/EFS and OS which has been explained in Section \ref{sec:pfs_os}, the general framework can be used beyond this example. For example, it might also be of interest to assess long-term efficacy (PFS) and long-term safety (as time to life-threatening toxicity or death). This results in a slightly more complex illness-death model as depicted by Figure \ref{fig:efficacy_safety_msm} and $k=3$, $d=2$, $E_1 = \{2,3\}$ (PFS) and $E_2 = \{1,3\}$ (safety).
\begin{figure}[h]
	\centering
	\begin{tikzpicture}
		\node[circle, draw] (initial) {$0$};
		\node[circle, draw] (progression) [above right=of initial] {$1$};
		\node[circle, draw] (toxicity) [below right=of initial] {$2$};
		\node[circle, draw] (death) [below right=of progression] {$3$};
		
		\draw [->] (initial) to (progression);
		\draw [->] (initial) to (toxicity);
		\draw [->] (initial) to (death);
		\draw [->] (progression) to[out = -75, in = 75] (toxicity);
		\draw [->] (progression) to (death);
		\draw [->] (toxicity) to[out = 105, in = -105] (progression);
		\draw [->] (toxicity) to (death);
		
		\begin{pgfonlayer}{background}
			\foreach \nodename in {progression, toxicity, death} {
				\coordinate (\nodename') at (\nodename);
			}
			\path[fill=blue!20, draw=green!20, line width = 1.2cm, line cap=round, line join=round] 
			(progression') to[bend left=0] (death') 
			to[bend left=0] (progression') -- cycle;
			\path[fill=red!20, draw=blue!20, line width = 1.0cm, line cap=round, line join=round] 
			(toxicity') to[bend left=0] (death') 
			to[bend left=0] (toxicity') -- cycle;
		\end{pgfonlayer}	
		
		\matrix [draw,below right,xshift = 30pt] at (current bounding box.north east) {
			\node [circle, draw, label=right:toxicity] {$1$}; \\
			\node [circle, draw, label=right:progression] {$2$}; \\
			\node [circle, draw, label=right:death] {$3$}; \\
			\node [circle, fill=green!20, label=right:toxic event] {$\quad$};\\
			\node [circle, fill=blue!20, label=right:PFS] {$\quad$};\\
		};
		
	\end{tikzpicture}
	\caption{Representation of the simultaneous assesment of efficacy and toxicity as a multi-state model} 
	\label{fig:efficacy_safety_msm}
\end{figure}
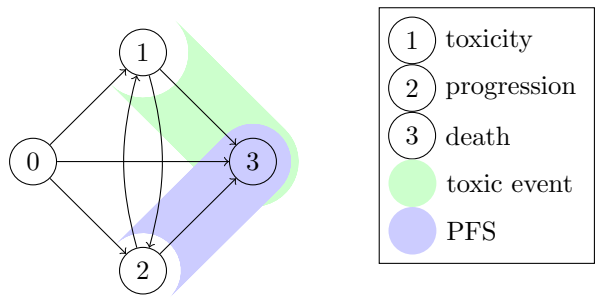
Additionally, it is notable, that this framework contains an adaptive design for a single-primary endpoint as a special case ($k=d=1$) and coincides with the procedure of \cite{Lin:1991} in a competing risks setting ($d = k$ and $E_c = \{c\}\, \forall c \in \{1,\dots,d\}$).

\section{Group sequential and adaptive designs}\label{sec:gs_design}
Based on the results obtained in the previous sections, we outline the procedure of a two-stage adaptive design for testing the null hypothesis \eqref{eq:H0_dist_fct_pfsos} resp. \eqref{eq:H0_cum_int_pfsos} in the illness-death model from Figure \ref{fig:pfsos}. Procedures with more than two stages and/or different endpoints can be constructed analogously. Stagewise $p$-values computed as in \eqref{eq:stagewise_p_value} can be used with any kind of combination function and decision boundary to set up a group sequential testing procedure. An extensive overview on these topics is e.g. given in \cite{Wassmer:2016}.\\
We do not consider stopping for futility here. Hence, an adaptive level $\alpha$ test can be specified by values $0 < \alpha_1 < \alpha < 1$ where $\alpha$ is the overall significance level and $\alpha_1$ the rejection level of the first stage, together with a conditional error function $\alpha_2\colon(\alpha_1,1] \to [0,1]$ which is monotonically decreasing and fulfills the equality
\begin{equation*}
	\int_{\alpha_1}^1 \alpha_2(x)\, dx = \alpha - \alpha_1.
\end{equation*}
In terms of the stagewise $p$-values, this leads to the rejection region $\{p_1 \leq \alpha_1\} \cup \{p_1 > \alpha_1, p_2 \leq \alpha_2(p_1)\}$. With regard to the stagewise increments of the multivariate process $\mathbf{U}$, $\alpha_1$ and $\alpha_2(p_1)$ induce an ellipse as the decision boundary (see Figure \ref{fig:nb2004}). At the interim analysis, the design of the trial (e.g. its sample size) may be adapted based on observations concerning PFS- and OS-events observed until this analysis date.

\subsection{Initial planing and sample size recalculations}
The second formulation \eqref{eq:H0_cum_int_pfsos} of the null hypothesis which is in terms of the transition intensities appears to be more natural as it directly refers to the components of the underlying model. In this way, treatment effects can easily be quantified via separate hazard ratios for each transition. Such a transition-wise consideration is presented in \cite{Le-Rademacher:2018} and employed for planning purposes in \cite{Erdmann:2023}. If the transition intensities in the control group ($Z=0$) and the treatment group ($Z=1$) are given by transition intensity functions $\lambda_0^{01}$, $\lambda_0^{02}$, $\lambda_0^{12}$ and $\lambda_1^{01}$, $\lambda_1^{02}$, $\lambda_1^{12}$, respectively, the transition-wise proportional hazard assumption amounts to
\begin{equation*}
	\frac{\lambda_1^{jk}(s)}{\lambda_0^{jk}(s)} = \delta^{jk} \qquad \forall\, s\geq 0\, ,(j,k) \in \{(0,1), (0,2), (1,2)\}
\end{equation*}
for some constants $\delta^{jk} > 0$. In particular, it follows that the porportional hazards assumption does in general not hold for OS or PFS. For PFS, for example, it can only be guaranteed if $\lambda_0^{01}$ and $\lambda_0^{02}$ are constant. Sample size calculation can also be adapted to cases in which transition intensities do not behave proportional.\\
It is shown in the Appendix, that in this case the increments $\mathbf{U}(t_r) - \mathbf{U}(t_{r-1})$ are approximately normally distributed with mean $\boldsymbol{\theta}(t_r) - \boldsymbol{\theta}(t_{r-1})$ governed by the size of the $\delta^{jk}$'s. The process $\boldsymbol{\theta}\coloneqq (\theta^{\text{PFS}}, \theta^{\text{OS}})$ is defined by
\begin{equation}\label{eq:pfs_drift}
	\theta^{\text{PFS}}(t)\coloneqq -\sqrt{n} \sum_{k=1}^2 (1 - \delta^{0k}) \int_{[0,t]} \left(1 - \frac{y^{0,Z=1}(t,u)}{y^{0}(t,u)} \right) y^{0,Z=1}(t,u) \lambda^{0k}_0(u) du
\end{equation}
\begin{equation}\label{eq:os_drift}
	\theta^{\text{OS}}(t)\coloneqq -\sqrt{n} \sum_{j=0}^1 (1 - \delta^{j2}) \int_{[0,t]} \left(1 - \frac{y^{j,Z=1}(t,u)}{y^{j}(t,u)} \right) y^{j,Z=1}(t,u) \lambda^{j2}_0(u) du.
\end{equation}
The covariance function $\mathbf{V}$ can be approximated as stated in the Appendix. In particular, the stagewise test statistic $S_m$ will asymptotically follow a non-central $\chi^2_2$-distribution with location parameter $n\cdot \eta_r$ with $\eta_r\coloneqq (\boldsymbol{\theta}(t_r) - \boldsymbol{\theta}(t_{r-1}))^T (\mathbf{V}(t_{r}) - \mathbf{V}(t_{r-1}))^{-1} (\boldsymbol{\theta}(t_r) - \boldsymbol{\theta}(t_{r-1}))$. All of the quantities mentioned here can be computed given a recruitment and censoring mechanism.\\
At the time of the interim analysis, each transition can be considered separately to check whether the original assumptions about the intensity functions are valid or need to be updated. Such updated assumptions may then be plugged into \eqref{eq:pfs_drift} and \eqref{eq:os_drift} to compute the conditional power of the procedure for the remaining stage. An adjustment of external parameters, such as the number of cases by extending the duration of recruitment, can then be made in order to achieve a desired conditional power. An example of such a procedure can be found in the Appendix. 

\section{Simulation study}\label{sec:simulation_study}
In this simulation study, we want to examine the procedure's compliance with the nominal type I error level as well as the sample size calculation procedure introduced in Section \ref{sec:gs_design} for practically relevant scenarios. 
\subsection{Design}
We consider different baseline scenarios of Markovian illness-death models presented in \cite{Meller:2019}. In each of the three scenarios the transition intensities are given by
\begin{equation}\label{eq:weibull_intensities}
	\lambda^{jk}(s) = \lambda^{jk} \cdot s^{\gamma^{jk} - 1}
\end{equation}
for any $(j,k) \in \{(0,1), (0,2), (1,2)\}$. The parameter values for the three scenarios considered here can be found in Table \ref{table:parameters}. The intensities thus have a Weibull form. In the special case in which $\gamma^{01} = \gamma^{02} = \gamma^{12} = 1$, the transition intensities are constant over time and the model is referred to as a time-homogenenous Markov model.\\
\begin{table}
	\centering
	\caption{Parameter configurations and event rates for the three scenarios of the simulation study}
	\begin{tabular}{|c|l|l|l|l|l|l|l|l|l|l|}
		\hline
		Scenario&$\gamma^{01}$&$\lambda^{01}$&$\gamma^{02}$&$\lambda^{02}$&$\gamma^{12}$&$\lambda^{12}$&$\pi^{\text{PFS}}(t_1)$&$\pi^{\text{OS}}(t_1)$&$\pi^{\text{PFS}}(t_2)$&$\pi^{\text{OS}}(t_2)$\\
		\hline
		1&1&0.6&1&0.075&1&0.9&0.431&0.241&0.889&0.745\\
		2&1.3&0.85&1.3&0.1&1.3&0.3&0.522&0.189&0.980&0.694\\
		3&1.5&0.57&0.5&0.065&0.85&1.1&0.441&0.235&0.957&0.772\\
		\hline
	\end{tabular}
	\label{table:parameters}
\end{table} 
For sample size calculation and type II error considerations, we assume that the groups differ in each transition by a proportionality factor as in Section \ref{sec:gs_design}. This means that the transition intensities in the experimental treatment group ($Z=1$) are given by
\begin{equation*}
	\lambda_1^{jk}(s) = \delta^{jk} \cdot \lambda_0^{jk} \quad \forall s\geq 0
\end{equation*}
given hazard ratios $\delta^{01}$, $\delta^{02}$ and $\delta^{12}$. In our scnearios, we assume values $\delta^{02} = 1$, $\delta^{01} \in \{0.8, 0.7, 0.6\}$ and $\delta^{12} \in \{0.85, 0.8, 0.75\}$.\\
We considered sequential designs with an interim analysis after $t_1 = 2.5$ and a final analysis at $t_2 = 5$ years. The duration of the accrual period was set to $a = 3$. The recruitment date of each trial participant was simulated as uniformly distributed on the interval $[0,a]$. Under these conditions, the expected proportion of all trial participants that will have experienced a PFS- or an OS-event by calendar time $t_1$ resp. $t_2$ are given in the last four columns of Table \ref{table:parameters} for each scenario.\\
The tests were carried out at an overall significance level of $\alpha = 5\%$. Stagewise $p$-values were combined using the inverse normal combination function with equal weights for the two stages. We applied the sequential decision boundaries according to Pocock as well as O'Brien-Fleming (abbreviated by P resp. OF in Tables \ref{table:emperrors} and \ref{table:sample_sizes}). We chose those designs as we consider them as well-known and -established in the wide range of group sequential and adaptive designs.\\
For type I error rate examination, we considered overall sample sizes $n \in \{50, 100, 250, 500, 1000\}$. For each alternative scenario, the respective analytically determined sample size to achieve a power of 80\% was used. The allocation ratio was always set to $1$, such that both groups are equal in size. For each constellation, 100,000 simulation runs were executed. Empirical rejection rates under the null hypothesis based on the two different sequential rejection boundaries were then computed and are displayed in Table \ref{table:emperrors}. Analytically computed sample sizes and empirical rejection rates for the different deviations from the null hypothesis are displayed in Table \ref{table:sample_sizes}.\\
The simulation study was performed with R 4.2.3 (see \cite{R}).

\subsection{Type I error rate}
\begin{table}
	\centering
	\caption{Empirical type I errors for different sample sizes, sequential decision bounds and scenarios}
	\begin{tabular}{|c|c|c|c|c|}
		\hline
		\multicolumn{2}{|c}{}&\multicolumn{3}{|c|}{Scenarios}\\
		\hline
		$n$&Type&1&2&3\\
		\hline
		\hline
		\multirow{2}{*}{50}&P&0.0531&0.0537&0.0557\\
		\cline{2-5}
		&OF&0.0546&0.0541&0.0566\\
		\hline
		\multirow{2}{*}{100}&P&0.0513&0.0520&0.0540\\
		\cline{2-5}
		&OF&0.0509&0.0528&0.0542\\
		\hline
		\multirow{2}{*}{250}&P&0.0505&0.0511&0.0521\\
		\cline{2-5}
		&OF&0.0512&0.0514&0.0517\\
		\hline
		\multirow{2}{*}{500}&P&0.0505&0.0511&0.0508\\
		\cline{2-5}
		&OF&0.0511&0.0516&0.0510\\
		\hline
		\multirow{2}{*}{1000}&P&0.0505&0.0489&0.0495\\
		\cline{2-5}
		&OF&0.0506&0.0500&0.0505\\
		\hline
	\end{tabular}
	\label{table:emperrors}
\end{table} 

For the chosen number of simulation runs, the 95\%-confidence interval for an underlying true value of 0.05 has a breadth of about 0.0027. Thus, it can be concluded that the presented method is slightly anti-conservative for small sample sizes (below 100). However, the deviation from the nominal significance level of 5\% is small. This tendency can also be observed for the standard log-rank test \cite{Heller:1996}, which can be regarded as a special case of our framework (with $k=1$, $d=1$ and $m=1$ in terms of the framework introduced in Section \ref{sec:general_framework}). The choice of sequential decision boundaries does not seem to play a role for the actually achieved significance level.

\subsection{Power and sample size}
\begin{table}
	\centering
	\caption{Empirically determined sample sizes for different sequential designs, scenarios and alternative hypotheses necessary to reach a power of 80\%}
	\begin{tabular}{|c|c|c|c|c|}
		\hline
		$\delta^{01}$&\multirow{2}{*}{Type}&\multicolumn{3}{|c|}{Scenarios}\\
		\cline{3-5}
		$\delta^{12}$&&1&2&3\\
		\hline
		\hline
		0.8&P&620 (79.88\%)&565 (80.27\%)&506 (80.12\%)\\
		\cline{2-5}
		0.85&OF&577 (79.98\%)&528 (80.17\%)&466 (80.16\%)\\
		\hline
		0.8&P&512 (80.01\%)&487 (80.10\%)&417 (80.02\%)\\
		\cline{2-5}
		0.8&OF&473 (79.83\%)&451 (80.00\%)&383 (79.97\%)\\
		\hline
		0.8&P&408 (79.72\%)&407 (79.79\%)&334 (79.74\%)\\
		\cline{2-5}
		0.75&OF&376 (80.15\%)&374 (79.99\%)&306 (79.95\%)\\
		\hline
		0.7&P&294 (80.25\%)&255 (80.11\%)&238 (80.06\%)\\
		\cline{2-5}
		0.85&OF&275 (80.11\%)&240 (80.02\%)&220 (80.12\%)\\
		\hline
		0.7&P&272 (80.00\%)&241 (80.17\%)&219 (80.25\%)\\
		\cline{2-5}
		0.8&OF&254 (80.25\%)&226 (80.03\%)&202 (80.08\%)\\
		\hline
		0.7&P&244 (79.98\%)&223 (80.12\%)&196 (79.98\%)\\
		\cline{2-5}
		0.75&OF&227 (80.01\%)&208 (80.05\%)&180 (79.83\%)\\
		\hline
		0.6&P&157 (80.04\%)&133 (80.20\%)&126 (80.13\%)\\
		\cline{2-5}
		0.85&OF&147 (80.21\%)&126 (80.22\%)&117 (80.23\%)\\
		\hline
		0.6&P&153 (80.33\%)&131 (80.30\%)&122 (80.25\%)\\
		\cline{2-5}
		0.8&OF&143 (80.22\%)&123 (79.93\%)&112 (80.11\%)\\
		\hline
		0.6&P&146 (80.24\%)&126 (79.99\%)&115 (80.01\%)\\
		\cline{2-5}
		0.75&OF&136 (80.32\%)&119 (80.38\%)&106 (79.97\%)\\
		\hline
	\end{tabular}
	\label{table:sample_sizes}
\end{table} 
For the chosen number of simulation runs, the 95\%-confidence interval for an underlying true value of 0.8 has a breadth of about 0.005. In view of the results from Table 2, it can therefore be assumed that the analytical determination of the number of cases described above works reliably in terms of compliance with the targeted power.\\
Even slight improvements concerning post-progression survival, which can be expressed by a smaller value of $\delta^{12}$ may substantially reduce the sample size required to achieve a power of 80\%. 

\section{Application example}\label{sec:application_example}
For further illustration of the methods introduced above, we now apply it to the data of the NB2004-HR trial (NCT number NCT03042429). This was an open-label, multicenter, prospective randomized controlled Phase III trial for treatment of children with High Risk Neuroblastoma. The patients received six (control intervention) resp. eight (experimental intervention) cycles of induction chemotherapy. Afterwards, both groups received the same high-dose chemotherapy with autologous stem cell rescue and a consolidation therapy afterwards (see \cite{Berthold:2020} for further details). The NB2004-HR trial had only one primary endpoint: Event-free survival (EFS), defined as time from diagnosis to progression, recurrence, secondary malignant disease or death, whatever occurs first. Nevertheless, post-progression survival is of key interest both here and in many other studies with EFS as primary endpoint as well. In particular, the interaction of first- and second-line therapy given after progression is of special interest. The analysis did not reveal a relevant difference between the two interventions, neither in the primary endpoint EFS nor in the secondary endpoint OS. To illustrate our methodology, we will reanalyse the NB2004-HR trial using our testing method as in the context of Section \ref{sec:pfs_os} in order to compare the joint distribution of EFS and OS between the two interventions.\\ 
The NB2004-HR trial was originally designed group-sequentially according to \cite{Pampallona:1994} including two interim analyses with futility stops and was later amended to an inverse normal adaptive design according to \cite{Wassmer:06} using the same rejection region as the initial group-sequential design. On this basis, a data-dependent sample size recalculation was performed at the second interim analysis. We mimic this design by conducting interim analyses at the same time points. However, we do not make any binding futility stops. Stage-wise decision boundaries are determined by adopting the alpha-spending that resulted from the original procedure. Stagewise $p$-values are combined using the inverse normal combination function with equal weights for all stages.\\
The results are displayed in Figure \ref{fig:nb2004}. Each of the three plots in Figure \ref{fig:nb2004} shows the value of the increment of $\sqrt{n}\mathbf{U}$ for the respective stage. As the test statistic $\sqrt{n}\mathbf{U}$ is bivariate, its observed value (as well as the corresponding rejection region) is located in the two-dimensional plane. The OS component is plotted in the vertical direction, the EFS component in the horizontal one. For both components of $\sqrt{n}\mathbf{U}$, negative values indicate an advantage of the experimental therapy in comparison with the control therapy. The red ellipses show the rejection bounds. If one of the test statistic increments would have been localised outside of the respective ellipse, the trial would have stopped with rejection of $H_0$. The exact shape of the ellipse that determines the rejection bound for the increments of $\sqrt{n}\mathbf{U}$ depends on the sequential decision boundaries in terms of $p$-values, the (estimated) variance of the increments of $\sqrt{n}\mathbf{U}$ given by the increments of $\hat{\mathbf{V}}$ as well as the results of previous analyses. The stagewise $p$-values resulting from our test turn out to be $p_1 = 0.536$, $p_2 = 0.227$ and $p_3 = 0.592$. Thus, the null hypothesis of no difference in the joint distribution of EFS and OS between the interventions cannot be rejected. This is qualitatively consistent with the results of the NB2004-HR trial as reported in \cite{Berthold:2020}.\\

\begin{figure}[h]
	\centering
	\includegraphics[width=\textwidth]{"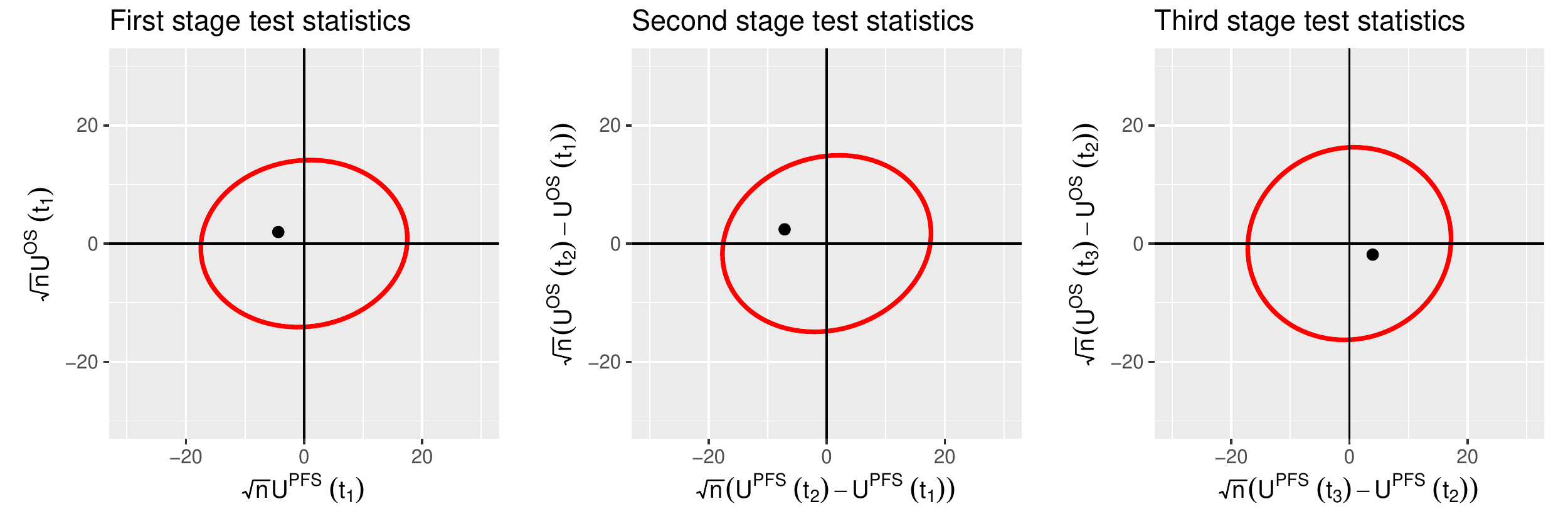"}
	\caption{Result of the evaluation of the NB2004-HR trial data with out new method}
	\label{fig:nb2004}
\end{figure}

In the original study, only primary outcome data on EFS were used for sample size recalculation as recommended by \cite{Wassmer:06}. The interim results from the first two phases suggested a slight benefit of the experimental treatment in terms of EFS, which is also evident from Figure \ref{fig:nb2004} in the form of a slight shift to the left of the observed statistic. This led the researchers to increase the number of events after which the final analysis should take place. This increase resulted from the requirement to achieve a conditional power of 80\% to reject the null hypothesis for EFS based on the original planning alternative. Information going beyond EFS-events has not been considered at the interim analyses. However, post-progression survival also plays a major role for a final assessment of a treatment for this disease. As one can see from the first two plots of Figure \ref{fig:nb2004}, a slightly unfavourable effect of the experimental treatment on post-progression survival has been observed at the interim analyses. This fact might have led the investigators to a different conclusion at the second interim analysis, if EFS and OS interim data had both been available in the context of the NB2004-HR trial.

\section{Discussion}\label{sec:discussion}
An adaptive group-sequential testing procedure for multiple primary time-to-event endpoints has been introduced. It serves as a generalisation to the adaptive log-rank test as presented e.g. by \cite{Wassmer:06} and coincides with the group-sequential procedure of \cite{Lin:1991} in case of a competing risks setting. As a consequence of the concerns raised in \cite{BPB:02}, an extension of \cite{Lin:1991} to an adaptive design is not straightforward. We do achieve that here by embedding these endpoints in a multi-state model under the assumption of Markovianity of this model. Our approach enables data-dependent interim design modifications based on the information about all involved endpoints. Similar to the one-sample procedure presented in \cite{Danzer:2022}, this is based on conditioning on the prior history of each patient which can be reduced to the current disease state under the Markov assumption.\\
As a particularly relevant application example from a practical point of view, we place a special focus on the joint consideration of PFS and OS in the framework of a simple illness-death model (see Figure \ref{fig:pfsos}). Both endpoints play a major role in oncology clinical trials. While OS is the most objectively defined endpoint, the choice of PFS as the primary endpoint is already established in many cases, depending on the specific indication. Often, as in our example in Section \ref{sec:application_example}, both endpoints are of crucial importance, suggesting a joint consideration of both. For immunotherapies in particular, it is possible that therapy effects only become apparent or even after a progression event \cite{Hoos:2012}. This is another reason why a joint consideration of the endpoints OS and PFS appears helpful. Using the data from the NB2004-HR study, we also show how the different aspects of our multivariate test can be visualised and interpreted (see Figure \ref{fig:nb2004}). The benefits that can be gained from our adaptive design in terms of interim, data-driven design changes have also been demonstrated through a case study in the Appendix.\\
Our simulation study has demonstrated that adherence to the nominal type I error rate is not only given asymptotically in the limit of large sample sizes, but is also acceptable at case numbers of practical relevance. We also considered effects of several differences in the survival pattern between the two groups on power and sample size of a corresponding study. In this regard, it should be noted that our procedure appears particularly suitable and superior to an adaptive test of the single endpoint PFS in terms of power if there is a relevant effect for post-progression survival. If no or only a very small effect with respect to OS is expected, obviously, the restriction to a classical adaptive test of the single endpoint PFS (in the sense of \cite{Wassmer:06}) appears more reasonable.\\
The methods presented here can be extended in several ways. To this end, it should be noted that the components of our general test statistic only take the first hitting time of some subset $E$ of the state space into account. Hence, it is not only a test for the null hypothesis \eqref{eq:H0_cum_int} formulated in terms of the cumulative intensity matrix, but also for the joint distribution of the $d$ different endpoints as in \eqref{eq:H0_dist_fct}. However, as an alternative to our approach, one could also think of test statistics that incorporate any hitting time of this set $E$ and not only the first one. The derivation of such a procedure is analogous to the derivation of the procedure on which we are focussing here. It is also carried out in full detail in the Appendix. These two approaches are the same for our illness-death model from Section \ref{sec:pfs_os} but can already differ for slightly more complex cases as e.g. the setting in Figure \ref{fig:efficacy_safety_msm}. However, the latter approach can only be used as a test for the null hypothesis \eqref{eq:H0_cum_int} as formulated in terms of the transition intensities.\\
Furthermore, we want to point out that we assumed transition-wise proportional hazards throughout our examples. Note that this generally does not imply proportional hazards for the endpoints (e.g. PFS and OS) considered within the multi-state model. In addition, settings are possible where the transition-wise comparisons may also not be subject to the proportional hazards assumption. If this is known, it might be beneficial to apply weights as it is also common for the univariate log-rank test (see e.g. Section V.2 in \cite{Andersen:93}). Such a weight can be selected separately for each individual transition. The theory lined out in the Appendix allows any weight fulfilling the standard assumptions.\\
An extension of our two-sample procedure to a $k$-sample procedure for some $k>2$ follows analogously to the way that the multivariate testing procedure from \cite{Wei:1984} is extended by \cite{Palesch:1994}, and is thus possible without further problems.\\
As stated in \eqref{eq:H0_dist_fct} resp. \eqref{eq:H0_cum_int}, the null hypothesis is formulated in a quite general way. Accordingly, we try to detect any kind of deviation from this hypothesis. In \cite{Wei:1984} and \cite{Lin:1991}, tests for one-sided hypotheses have also been suggested which could be adopted. However, these may be sensitive to undesirable alternatives as demonstrated in \cite{Bloch:2001}. In any case, rejection of the general null hypothesis should be followed by more in-depth analyses. This could be achieved by a closed testing procedure involving the various components of $\mathbf{M}$, similar to the suggestions made in \cite{Lehmacher:1991}. A separate analysis of the transition intensities as demonstrated in Section IV.4.4 of \cite{Andersen:93} is also recommendable.
The correctness of the procedure requires the Markov assumption. This allows us to adequately incorporate the information gathered so far into the testing procedure. Before use, the appropriateness of this assumption should therefore be investigated. On the one hand, this can be based on the expertise of clinical investigators. On the other hand, it can be examined in historical data sets that reflect the population to be recruited for the present trial. Corresponding testing procedures have been developed for the simple illness death model of Figure \ref{fig:pfsos} in \cite{Rodriguez:2012} as well as for general multi-state models in \cite{Titman:2020}.\\
Considering the topics discussed here, we strive to further develop and improve our framework in future research to enhance applicability in clinical trials. In principle, analogous methods can be developed in non-Markov settings, e.g. in the scenario of Semi-Markov models (see e.g. \cite{Meller:2019} for details on the Semi-Markovia illness-death model). Furthermore, we aim to develop tests for marginal distributions of time-to-event endpoints in Markovian multi-stage models. Compared to the current methodology, these should not only consider the conditional distribution of an endpoint and still allow for the possibility of interim design adaptations based on the disease history data of all patients.

\section*{Acknowledgement}
This work was funded by the Deutsche Forschungsgemeinschaft (DFG, German Research Foundation) – 413730122. The NB2004-HR trial was funded by a grant obtained by Frank Berthold from the Deutsche Krebshilfe (grant number 70107712). We thank Frank Berthold for the opportunity to scientifically use the anonymized EFS and OS data of the NB2004-HR trial (published elsewhere) as a real clinical example in the context of this work.

\bibliographystyle{unsrt}
\bibliography{manuscript_mvlr_bib}

\begin{thebibliography}{10}

\bibitem{Muller:01}
Helmut Schäfer and Hans-Helge Müller.
\newblock Modification of the sample size and the schedule of interim analyses
  in survival trials based on data inspections.
\newblock {\em Statistics in Medicine}, 20(24):3741--3751, 2001.

\bibitem{Wassmer:06}
Gernot Wassmer.
\newblock Planning and analyzing adaptive group sequential survival trials.
\newblock {\em Biometrical Journal}, 48(4):714--729, 2006.

\bibitem{Tsiatis:1981}
Anastasios~A. Tsiatis.
\newblock The asymptotic joint distribution of the efficient scores test for
  the proportional hazards model calculated over time.
\newblock {\em Biometrika}, 68(1):311--315, 1981.

\bibitem{Sellke82}
Thomas Sellke and David Siegmund.
\newblock Sequential analysis of the proportional hazards model.
\newblock {\em Biometrika}, 70(2):315--326, 08 1983.

\bibitem{BPB:02}
Peter Bauer and Martin Posch.
\newblock Letter to the editor: Modification of the sample size and the
  schedule of interim analyses in survival trials based on data inspections.
\newblock {\em Statistics in Medicine}, 23(8):1333--1334, 2004.

\bibitem{Jenkins:11}
Martin Jenkins, Andrew Stone, and Christopher Jennison.
\newblock An adaptive seamless phase ii/iii design for oncology trials with
  subpopulation selection using correlated survival endpoints†.
\newblock {\em Pharmaceutical Statistics}, 10(4):347--356, 2011.

\bibitem{Irle:12}
Sebastian Irle and Helmut Schäfer.
\newblock Interim design modifications in time-to-event studies.
\newblock {\em Journal of the American Statistical Association},
  107(497):341--348, 2012.

\bibitem{Jörgens:2019}
Silke Jörgens, Gernot Wassmer, Franz König, and Martin Posch.
\newblock Nested combination tests with a time-to-event endpoint using a
  short-term endpoint for design adaptations.
\newblock {\em Pharmaceutical Statistics}, 18(3):329--350, 2019.

\bibitem{Magirr:16}
Dominic Magirr, Thomas Jaki, Franz Koenig, and Martin Posch.
\newblock Sample size reassessment and hypothesis testing in adaptive survival
  trials.
\newblock {\em PLOS ONE}, 11(2):1--14, 02 2016.

\bibitem{Danzer:2022}
Moritz~F. Danzer, Tobias Terzer, Frank Berthold, Andreas Faldum, and Rene
  Schmidt.
\newblock Confirmatory adaptive group sequential designs for single-arm phase
  ii studies with multiple time-to-event endpoints.
\newblock {\em Biometrical Journal}, 64(2):312--342, 2022.

\bibitem{Wei:1984}
Lee~J. Wei and John~M. Lachin.
\newblock Two-sample asymptotically distribution-free tests for incomplete
  multivariate observations.
\newblock {\em Journal of the American Statistical Association},
  79(387):653--661, 1984.

\bibitem{Lin:1991}
Danyu Lin.
\newblock Nonparametric sequential testing in clinical trials with incomplete
  multivariate observations.
\newblock {\em Biometrika}, 78(1):123--131, 1991.

\bibitem{Le-Rademacher:2018}
Jennifer Le-Rademacher, Ryan Peterson, Terry Therneau, Ben Sanford, Richard
  Stone, and Sumithra Mandrekar.
\newblock Application of multi-state models in cancer clinical trials.
\newblock {\em Clinical Trials}, 15:174077451878909, 07 2018.

\bibitem{Bellera:2013}
Carine~A. Bellera, Marina Pulido, Sophie Gourgou, Laurence Collette, Adélaïde
  Doussau, Andrew Kramar, Tienhan~Sandrine Dabakuyo, Monia Ouali, Anne Auperin,
  Thomas Filleron, Catherine Fortpied, Christophe {Le Tourneau}, Xavier
  Paoletti, Murielle Mauer, Simone Mathoulin-Pélissier, and Franck Bonnetain.
\newblock Protocol of the {D}efinition for the {A}ssessment of {T}ime-to-event
  {E}ndpoints in {CAN}cer trials ({DATECAN}) project: Formal consensus method
  for the development of guidelines for standardised time-to-event endpoints’
  definitions in cancer clinical trials.
\newblock {\em European Journal of Cancer}, 49(4):769--781, 2013.

\bibitem{Meller:2019}
Matthias Meller, Jan Beyersmann, and Kaspar Rufibach.
\newblock Joint modeling of progression-free and overall survival and
  computation of correlation measures.
\newblock {\em Statistics in Medicine}, 38(22):4270--4289, 2019.

\bibitem{Tattar:2014}
Prabhanjan~N. Tattar and H.~J. Vaman.
\newblock The k-sample problem in a multi-state model and testing transition
  probability matrices.
\newblock {\em Lifetime Data Analysis}, 20(3):387--403, 2014.

\bibitem{Wassmer:2016}
Gernot Wassmer and Werner Brannath.
\newblock {\em Group Sequential and Confirmatory Adaptive Designs in Clinical
  Trials}.
\newblock Springer Series in Pharmaceutical Statistics. Springer International
  Publishing, 2016.

\bibitem{Erdmann:2023}
Alexandra Erdmann, Jan Beyersmann, and Kaspar Rufibach.
\newblock Oncology clinical trial design planning based on a multistate model
  that jointly models progression-free and overall survival endpoints, 2023.
\newblock arXiv:2301.10059.

\bibitem{R}
{R Core Team}.
\newblock {\em R: A Language and Environment for Statistical Computing}.
\newblock R Foundation for Statistical Computing, Vienna, Austria, 2014.

\bibitem{Heller:1996}
Glenn Heller and E.~S. Venkatraman.
\newblock Resampling procedures to compare two survival distributions in the
  presence of right-censored data.
\newblock {\em Biometrics}, 52(4):1204--1213, 1996.

\bibitem{Berthold:2020}
Frank Berthold, Andreas Faldum, Angela Ernst, Joachim Boos, Dagmar Dilloo,
  Angelika Eggert, Matthias Fischer, Michael Frühwald, Günter Henze, Thomas
  Klingebiel, Christian Kratz, Bernhard Kremens, Barbara Krug, Ivo Leuschner,
  Matthias Schmidt, Rene Schmidt, Roswitha Schumacher-Kuckelkorn, Dietrich von
  Schweinitz, Freimut~H. Schilling, Jessica Theissen, Ruth Volland, Barbara
  Hero, and Thorsten Simon.
\newblock Extended induction chemotherapy does not improve the outcome for
  high-risk neuroblastoma patients: results of the randomized open-label gpoh
  trial {NB2004-HR}.
\newblock {\em Annals of Oncology}, 31(3):422 -- 429, 2020.

\bibitem{Pampallona:1994}
Sandro Pampallona and Anastasios~A. Tsiatis.
\newblock Group sequential designs for one-sided and two-sided hypothesis
  testing with provision for early stopping in favor of the null hypothesis.
\newblock {\em Journal of Statistical Planning and Inference}, 42(1):19--35,
  1994.

\bibitem{Hoos:2012}
Axel Hoos.
\newblock Evolution of end points for cancer immunotherapy trials.
\newblock {\em Annals of Oncology}, 23:viii47--viii52, 2012.
\newblock Advances in Immuno-oncology.

\bibitem{Andersen:93}
Per~Kragh Andersen, Ørnulf Borgan, Richard~D. Gill, and Niels Keiding.
\newblock {\em Statistical Models Based on Counting Processes}.
\newblock Springer, 1993.

\bibitem{Palesch:1994}
Yuko~Y. Palesch and John~M. Lachin.
\newblock Asymptotically distribution-free multivariate rank tests for multiple
  samples with partially incomplete observations.
\newblock {\em Statistica Sinica}, 4(1):373--387, 1994.

\bibitem{Bloch:2001}
Daniel~A. Bloch, Tze~Leung Lai, and Pascale Tubert-Bitter.
\newblock One-sided tests in clinical trials with multiple endpoints.
\newblock {\em Biometrics}, 57(4):1039--1047, 2001.

\bibitem{Lehmacher:1991}
Walter Lehmacher, Gernot Wassmer, and Peter Reitmeir.
\newblock Procedures for two-sample comparisons with multiple endpoints
  controlling the experimentwise error rate.
\newblock {\em Biometrics}, 47(2):511--521, 1991.

\bibitem{Rodriguez:2012}
Mar Rodríguez-Girondo and Jacobo de~Uña-Álvarez.
\newblock A nonparametric test for markovianity in the illness-death model.
\newblock {\em Statistics in Medicine}, 31(30):4416--4427, 2012.

\bibitem{Titman:2020}
Andrew~C. Titman and Hein Putter.
\newblock {General tests of the Markov property in multi-state models}.
\newblock {\em Biostatistics}, 23(2):380--396, 09 2020.

\bibitem{Puri:1984}
Madan~L. Puri, Carl~T. Russell, and Thomas Mathew.
\newblock Convergence of generalized inverses with applications to asymptotic
  hypothesis testing.
\newblock {\em Sankhyā: The Indian Journal of Statistics, Series A
  (1961-2002)}, 46(2):277--286, 1984.

\bibitem{Wanek:1993}
Leslie~A. Wanek, Tushar~M. Goradia, Robert~M. Elashoff, and Donald~L. Morton.
\newblock Multi-stage markov analysis of progressive disease applied to
  melanoma.
\newblock {\em Biometrical Journal}, 35(8):967--983, 1993.

\bibitem{Fleischer:2009}
Frank Fleischer, Birgit Gaschler-Markefski, and Erich Bluhmki.
\newblock A statistical model for the dependence between progression-free
  survival and overall survival.
\newblock {\em Statistics in Medicine}, 28(21):2669--2686, 2009.

\end{thebibliography}

\section*{Appendix}

\subsection*{A.1 Proofs of main results}

In what follows we will prove the asymptotic results mentioned in the main manuscript. To do so, we will consider the different transitions of the Markovian multi-state model individually and use techniques as in the univariate setting (see e.g. \cite{Tsiatis:1981}). Afterwards, the results thus obtained can be plugged together.\\
For two states $j\neq k$ of the Markovian multi-state model, we will consider the collection of $\sigma$-algebras $(\mathcal{F}^{\{j\}\{k\}}_i(t,s))_{t,s \geq 0}$ containing only information about the transition $j \to k$. The $\sigma$-algebra $\mathcal{F}^{\{j\}\{k\}}_{i}(t,s)$ is generated by the random variables
\begin{equation*}
	\begin{split}
		&\mathbbm{1}_{\{R_i \leq t\}}, R_i \cdot \mathbbm{1}_{\{R_i \leq t\}}, \mathbbm{1}_{\{\tilde{C}_i 	\leq s \wedge (t-R_i)_+\}}, \tilde{C}_i \cdot \mathbbm{1}_{\{\tilde{C}_i \leq s \wedge (t-R_i)_+\}},\{ Y_i^j(t,u) | u \leq s \},\\
		&\mathbbm{1}_{\{T_i^{\{k\}} \leq s \wedge C_i(t)\}} \cdot \mathbbm{1}_{\{ X_i(T_i^{\{k\}}-) = j \}}, T_i^{\{k\}}\cdot \mathbbm{1}_{\{T_i^{\{k\}} \leq s \wedge C_i(t)\}} \cdot \mathbbm{1}_{\{ X_i(T_i^{\{k\}}-) = j \}}
	\end{split}
\end{equation*}
and hence contains any information about the transition $j \to k$ of patient $i$ available until calendar time $t$ and trial time $s$. The process $(N_i^{\{j\}\{k\}}(t,s))_{t,s \geq 0} $, defined by
\begin{equation*}
	N_i^{\{j\}\{k\}}(t,s) \coloneqq \mathbbm{1}_{\{T_i^{\{k\}} \leq s \wedge C_i(t)\}} \cdot \mathbbm{1}_{\{X_i(T_i^{\{k\}}-) = j\}},
\end{equation*}
indicates whether a transition $j \to k$ has occured for patient $i$ before calendar time $t$ and trial time $s$. Then the compensated process $(M_i^{\{j\}\{k\}}(t,s))_{s \geq 0}$ with
\begin{equation}\label{eq:transition_specific_martingale}
	M_i^{\{j\}\{k\}}(t,s) \coloneqq N_i^{\{j\}\{k\}}(t,s) - \int_{[0,s]} Y_i^j(t,u) \lambda^{jk}(u) \, du
\end{equation}
is known to be a martingale with respect to (w.r.t.) the filtration $(\mathcal{F}^{\{j\}\{k\}}_{i}(t,s))_{s \geq 0}$ for any fixed $t \geq 0$. The aggregated information concerning all patients of the study sample is thus given by the collection of $\sigma$-algebras $(\mathcal{F}^{\{j\}\{k\}}(t,s))_{t,s \geq 0}$ defined by $\mathcal{F}^{\{j\}\{k\}}(t,s) \coloneqq \sigma(\cup_{i=1}^n \mathcal{F}^{\{j\}\{k\}}_i(t,s))$. Based on that, we further define the collections of $\sigma$-algebras $(\mathcal{F}^{\{k\}}(t,s))_{t,s \geq 0}$, $(\mathcal{F}^{E}(t,s))_{t,s \geq 0}$ and $(\mathcal{F}(t,s))_{t,s \geq 0}$ given by
\begin{equation*}
	\mathcal{F}^{\{k\}}(t,s)\coloneqq \sigma\left( \cup_{j \in \{0,\dots,l\} \colon \lambda^{jk}\not\equiv 0} \mathcal{F}^{\{j\}\{k\}}(t,s) \right)
\end{equation*}
resp.
\begin{equation*}
	\mathcal{F}^{E}(t,s)\coloneqq \sigma\left( \cup_{k \in E} \mathcal{F}^{\{k\}}(t,s) \right)
\end{equation*}
resp.
\begin{equation*}
	\mathcal{F}(t,s)\coloneqq \sigma\left( \cup_{(j,k) \colon \lambda^{jk} \not\equiv 0} \mathcal{F}^{\{j\}\{k\}}(t,s) \right)
\end{equation*}
for any $k \in \{1,\dots,l\}$, $E \subset \{1,\dots,l\}$ and $s,t\geq 0$.\\
We will consider first the asymptotic behaviour of the multivariable transition specific stochastic processes $(U^{\{j\}\{k\}}(t,s))_{t,s \geq 0}$ and $(U^{\{j\}\{k\}}(t))_{t \geq 0}$ given by
\begin{equation*}\label{eq:transition_stat}
	U^{\{j\}\{k\}}(t,s)\coloneqq \frac{1}{\sqrt{n}} \sum_{i=1}^n \int_{[0,s]} Q^{jk}(t,u) \left( Z_i - \frac{Y^{j, Z=1}(t,u)}{Y^{j}(t,u)} \right)\, N_i^{\{j\}\{k\}}(t, du)
\end{equation*}
resp. 
\begin{equation*}
	U^{\{j\}\{k\}}(t)\coloneqq U^{\{j\}\{k\}}(t,t)
\end{equation*}
with a random weight function $(Q^{jk}(t,s))_{t,s \geq 0}$ that is predictable as a stochastic process in its second argument for any $t \geq 0$ w.r.t. $(\mathcal{F}(t,s))_{s \geq 0}$.
\begin{thm}\label{thm:transition_stat}
	Let $j \neq k$ be two states of a Markovian multi-state model with $\lambda^{jk} \not\equiv 0$. We define the stochastic processes $(u^{\{j\}\{k\}}(t,s))_{t,s \geq 0}$ and $(u^{\{j\}\{k\}}(t))_{t\geq 0}$ by
	\begin{equation}\label{eq:transition_stat_limit}
		u^{\{j\}\{k\}}(t,s)\coloneqq \frac{1}{\sqrt{n}} \sum_{i=1}^n \underbrace{ \int_{[0,s]} q^{jk}(t,u) \left( Z_i - \frac{y^{j, Z=1}(t,u)}{y^{j}(t,u)} \right)\, M_i^{\{j\}\{k\}}(t, du) }_{ \eqqcolon u_i^{\{j\}\{k\}}(t,s) }
	\end{equation}
	resp. 
	\begin{equation*}
		u^{\{j\}\{k\}}(t)\coloneqq u^{\{j\}\{k\}}(t,t).
	\end{equation*}
	Under the assumptions
	\begin{enumerate}[label = (A\arabic*)]
		\item \label{item:assumption:weight_convergence}
		\begin{equation*}
			\sup_{0\leq s \leq \tau} |Q^{jk}(t,s) - q^{jk}(t,s)| \overset{\mathbb{P}}{\to}0
		\end{equation*}
		for any $\tau < t$,
		\item \label{item:assumption:at_risk}
		\begin{equation*}
			\sup_{\tau_1 \leq s \leq \tau_2} \left| \frac{y^{j, Z=1}(t,s)}{y^{j}(t,s)}- \frac{Y^{j, Z=1}(t,s)}{Y^{j}(t,s)} \right| \overset{\mathbb{P}}{\to} 0
		\end{equation*} 
		for any $\tau_1 > 0$, $\tau_2 < t$ where $y^{j}(t,s) = \mathbb{E}[Y^{j}(t,s)]$ and $y^{j, Z=1}(t,s) = \mathbb{E}[Y^{j, Z=1}(t,s)]$,
		\item \label{item:assumption:weight_bound} $Q^{jk}(t,s)$ is bounded in its second argument over $[0,t]$, left continuous and with right hand limits,
	\end{enumerate}
	the processes $(U^{\{j\}\{k\}}(t))_{t\geq 0}$ and $(u^{\{j\}\{k\}}(t))_{t\geq 0}$ are asymptotically equivalent, i.e.
	\begin{equation*}
		(U^{\{j\}\{k\}}(t) - u^{\{j\}\{k\}}(t)) \overset{\mathbb{P}}{\to} 0
	\end{equation*}
	for any $t \geq 0$.
\end{thm}
\begin{proof}
	First, we note that for any $t\geq 0$
	\begin{equation*}
		U^{\{j\}\{k\}}(t) = \frac{1}{\sqrt{n}} \sum_{i=1}^n \int_{[0,t]} Q^{jk}(t,u) \left( Z_i - \frac{Y^{j, Z=1}(t,u)}{Y^{j}(t,u)} \right)\, M_i^{\{j\}\{k\}}(t, du)
	\end{equation*}
	as
	\begin{align*}
		&U^{\{j\}\{k\}}(t) - \frac{1}{\sqrt{n}} \sum_{i=1}^n \int_{[0,t]} Q^{jk}(t,u) \left( Z_i - \frac{Y^{j, Z=1}(t,u)}{Y^{j}(t,u)} \right)\, M_i^{\{j\}\{k\}}(t, du)\\
		=&\frac{1}{\sqrt{n}} \sum_{i=1}^n \int_{[0,t]} Q^{jk}(t,u) \left( Z_i - \frac{Y^{j,Z=1}(t,u)}{Y^j(t,u)} \right) d\left( \int_{[0,u]} Y_i^j(t,v) \lambda^{jk}(v) dv \right)\\
		=&\frac{1}{\sqrt{n}} \sum_{i=1}^n \int_{[0,t]} Q^{jk}(t,u) Z_i Y_i^j(t,u) \lambda^{jk}(u)\,du\\
		&-\frac{1}{\sqrt{n}} \sum_{i=1}^n \int_{[0,t]} Q^{jk}(t,u) \frac{Y^{j,Z=1}(t,u)}{Y^j(t,u)} Y_i^j(t,u) \lambda^{jk}(u)\,du\\
		=&\frac{1}{\sqrt{n}} \int_{[0,t]} Q^{jk}(t,u) Y^{j, Z=1}(t,u) \lambda^{jk}(u)\,du - \frac{1}{\sqrt{n}} \int_{[0,t]} Q^{jk}(t,u) Y^{j}(t,u) \frac{Y^{j,Z=1}(t,u)}{Y^j(t,u)} \lambda^{jk}(u)\,du\\
		=&0.
	\end{align*}
	Now, we can rearrange the difference between the two processes
	\begin{align*}
		&U^{\{j\}\{k\}}(t) - u^{\{j\}\{k\}}(t)\\
		=&\frac{1}{\sqrt{n}} \sum_{i=1}^n \int_{[0,t]} (Q^{jk}(t,u) - q^{jk}(t,u)) \left( Z_i - \frac{Y^{j, Z=1}(t,u)}{Y^{j}(t,u)} \right) M_i^{\{j\}\{k\}}(t,du)\\
		&+\frac{1}{\sqrt{n}} \sum_{i=1}^n \int_{[0,t]} q^{jk}(t,u) \left( \frac{y^{j, Z=1}(t,u)}{y^{j}(t,u)} - \frac{Y^{j, Z=1}(t,u)}{Y^{j}(t,u)} \right) M_i^{\{j\}\{k\}}(t,du).
	\end{align*}
	As $Z_i$ takes values in $\{0,1\}$, we can split the two groups by
	\begin{align*}
		&U^{\{j\}\{k\}}(t) - u^{\{j\}\{k\}}(t)\\
		=&\sum_{g=0}^1 \frac{1}{\sqrt{n}} \int_{[0,t]} (Q^{jk}(t,u) - q^{jk}(t,u)) \left( g - \frac{Y^{j, Z=1}(t,u)}{Y^{j}(t,u)} \right) \bar{M}_g^{\{j\}\{k\}}(t,du)\\
		&+ \sum_{g=0}^1 \frac{1}{\sqrt{n}}  \int_{[0,t]} q^{jk}(t,u) \left( \frac{y^{j, Z=1}(t,u)}{y^{j}(t,u)} - \frac{Y^{j, Z=1}(t,u)}{Y^{j}(t,u)} \right) \bar{M}_g^{\{j\}\{k\}}(t,du) .
	\end{align*}
	where $\bar{M}_g^{\{j\}\{k\}}(t,s)\coloneqq \sum_{i=1}^n M_i^{\{j\}\{k\}}(t,s) \cdot \mathbbm{1}_{Z_i = g}$ are the group specific aggregates of the transition specific martingales. The integrands of both terms are bounded and predictable w.r.t. the filtration $(\mathcal{F}^{\{j\}\{k\}}_{t,s})_{0 \leq s \leq t}$ and $(\bar{M}_g^{\{j\}\{k\}})_{0\leq s \leq t}$ is a mean zero and square integrable martingale w.r.t. the same filtration. Hence, the processes
	$(A_{g,t}(s))_{0\leq s \leq t}$ and $(B_{g,t}(s))_{0\leq s \leq t}$ with $g \in \{0,1\}$ defined by
	\begin{equation*}
		A_{g,t}(s)\coloneqq \frac{1}{\sqrt{n}} \int_{[0,s]} (Q^{jk}(t,u) - q^{jk}(t,u)) \left( g - \frac{Y^{j, Z=1}(t,u)}{Y^{j}(t,u)} \right) \bar{M}_g^{\{j\}\{k\}}(t,du)
	\end{equation*}
	and
	\begin{equation*}
		B_{g,t}(s)\coloneqq \frac{1}{\sqrt{n}} \int_{[0,s]} q^{jk}(t,u) \left( \frac{y^{j, Z=1}(t,u)}{y^{j}(t,u)} - \frac{Y^{j, Z=1}(t,u)}{Y^{j}(t,u)} \right) \bar{M}_g^{\{j\}\{k\}}(t,du)
	\end{equation*}
	are also mean zero and square integrable martingales with predictable covariation processes
	\begin{equation*}
		\langle A_{g,t} \rangle (s) = \frac{1}{n} \int_{[0,s]} (Q^{jk}(t,u) - q^{jk}(t,u))^2 \left( g - \frac{Y^{j, Z=1}(t,u)}{Y^{j}(t,u)} \right)^2 \lambda^{jk}(u) Y^{j,Z=1}(t,u)\,du
	\end{equation*}
	resp.
	\begin{equation*}
		\langle B_{g,t} \rangle (s) = \frac{1}{n} \int_{[0,s]} q^{jk}(t,u)^2 \left( \frac{y^{j, Z=1}(t,u)}{y^{j}(t,u)} - \frac{Y^{j, Z=1}(t,u)}{Y^{j}(t,u)} \right)^2 \lambda^{jk}(u) Y^{j,Z=1}(t,u)\,du.
	\end{equation*}
	For any fixed $\varepsilon > 0$ we can find some $\tau_{\varepsilon} > 0$ s.t. 
	\begin{equation*}
		\int_{[s-\tau_{\varepsilon}, s]} \lambda^{jk}(u)\, du \cdot \sup_{0\leq u \leq t} \left((Q^{jk}(t,u) - q^{jk}(t,u))^2 \left( g - \frac{Y^{j, Z=1}(t,u)}{Y^{j}(t,u)} \right)^2\right) < \frac{\varepsilon}{2}
	\end{equation*}
	resp.
	\begin{equation*}
		\int_{[0,\tau_{\varepsilon}]\cup[s-\tau_{\varepsilon},s]} \lambda^{jk}(u)\, du \cdot \sup_{0\leq u \leq t} \left(q^{jk}(t,u)^2 \left( \frac{y^{j, Z=1}(t,u)}{y^{j}(t,u)} - \frac{Y^{j, Z=1}(t,u)}{Y^{j}(t,u)} \right)^2\right) < \frac{\varepsilon}{2}
	\end{equation*}
	$\mathbb{P}$-almost surely for any $0\leq s \leq t$. We can find such a value because $\lambda^{jk}$ is bounded on compact intervals and hence the supremum is taken over values which are all bounded by definition or the additional assumptions. Further, since
	\begin{equation*}
		\sup_{0 \leq u \leq t-\tau_{\varepsilon}}|Q^{jk}(t,u) - q^{jk}(t,u)|\overset{\mathbb{P}}{\to}0\quad \text{and} \quad \sup_{\tau_{\varepsilon} \leq u \leq t-\tau_{\varepsilon}}\left| \frac{y^{j, Z=1}(t,u)}{y^{j}(t,u)} - \frac{Y^{j, Z=1}(t,u)}{Y^{j}(t,u)} \right| \overset{\mathbb{P}}{\to}0
	\end{equation*}
	for any $\delta > 0$ we can also find some $n_0$ large enough s.t. for any $n\geq n_0$,
	\begin{equation*}
		\begin{split}
			&\mathbb{P} \left[ \Lambda^{jk}(t) \sup_{0 \leq u \leq t-\tau_{\varepsilon}}|Q^{jk}(t,u) - q^{jk}(t,u)| < \frac{\varepsilon}{2} \right] > 1-\delta \qquad \text{and}\\
			&\mathbb{P} \left[ \Lambda^{jk}(t) \sup_{\tau_{\varepsilon} \leq u \leq t-\tau_{\varepsilon}}\left| \frac{y^{j, Z=1}(t,u)}{y^{j}(t,u)} - \frac{Y^{j, Z=1}(t,u)}{Y^{j}(t,u)} \right| < \frac{\varepsilon}{2} \right] > 1-\delta.
		\end{split}
	\end{equation*}
	Consequently, 
	\begin{equation*}
		\mathbb{P}\left[ \langle A_{g,t} \rangle (s) < \varepsilon \right] > 1-\delta \quad \text{and} \quad \mathbb{P}\left[ \langle B_{g,t} \rangle (s) < \varepsilon \right] > 1-\delta
	\end{equation*}
	for any $0\leq s \leq t$, $n\geq n_0$ and $g \in \{0,1\}$. Following Theorem II.5.1 of \cite{Andersen:93}, we obtain in particular
	\begin{equation*}
		(U^{\{j\}\{k\}}(t) - u^{\{j\}\{k\}}(t)) \overset{\mathbb{P}}{\to} 0
	\end{equation*}
	as $n \to \infty$.
\end{proof}

We will use this result as a basis to construct processes that do not only sum up the information about a single transition $j\to k$. As we are mostly interested in the first hitting time of a subset $E$ of the state space of our model as the time of a corresponding composite event, we will need to stop the observation of patients which have already experiences some event. To make this more formal, we will consider some collection of random variables $(V_i(t))_{i \in \{1,\dots,n\}, t\geq 0}$ s.t. $V_i(t)$ is a $(\mathcal{F}(t,s))_{s \geq 0}$-stopping time for any $i \in \{1,\dots,n\}$ and $t \geq 0$ and define the stopped processes $(\tilde{M}^{\{j\}\{k\}}_i(t,s))_{t,s \geq 0}$ by
\begin{align*}
	\tilde{M}^{\{j\}\{k\}}_i(t,s) & \coloneqq M^{\{j\}\{k\}}_i(t,s \wedge V_i(t))\\
	& = \underbrace{N^{\{j\}\{k\}}_i(t,s) \cdot \mathbbm{1}_{\{T_i^{\{k\}} \leq V_i(t)\}}}_{\eqqcolon \tilde{N}^{\{j\}\{k\}}_i(t,s)} - \int_{[0,s]} \underbrace{Y_i^{j}(t,u) \cdot \mathbbm{1}_{\{u \leq V_i(t)\}}}_{\eqqcolon \tilde{Y}_i^{j}(t,u)} \lambda^{jk}(u)\, du  
\end{align*}
Accordingly, we need to consider the adjusted "at risk" processes $(\tilde{Y}^{j}(t,s))_{t,s \geq 0}$ and $(\tilde{Y}^{j,Z=1}(t,s))_{t,s \geq 0}$ by accumulating the corresponding patient-wise risk indicators, i.e.
\begin{equation*}
	\tilde{Y}^{j}(t,s)\coloneqq \sum_{i=1}^n \tilde{Y}_i^{j}(t,s) \qquad \text{resp.} \qquad \tilde{Y}^{j, Z=1}(t,s)\coloneqq \sum_{i=1}^n \tilde{Y}_i^{j}(t,s) \cdot \mathbbm{1}_{\{Z_i = 1\}}.
\end{equation*}
Furthermore, we define the processes $(\tilde{U}^{\{j\}\{k\}}(t,s))_{t,s \geq 0}$ and $((\tilde{u}^{\{j\}\{k\}}(t,s))_{t,s \geq 0})$ by 
\begin{enumerate}
	\item replacing $Y$ and $y$ in \eqref{eq:transition_stat} and \eqref{eq:transition_stat_limit} by $\tilde{Y}$ and $\tilde{y}$, respectively, and
	\item stopping the integrator process at $V_i(t)$ (or equivalently integrating over $[0, s \wedge V_i(t)]$ instead of $[0,s]$).
\end{enumerate}
\begin{cor}\label{cor:transition_stat_stopped}
	Let the conditions of Theorem \ref{thm:transition_stat} be given where in \ref{item:assumption:at_risk} the "at risk"-processes $Y$ and $y$ are replaced by $\tilde{Y}$ and $\tilde{y}$, respectively. Let $(V_i(t))_{i \in \{1,\dots,n\}, t \geq 0}$ be a collection of random variables s.t. $V_i(t)$ is a  $(\mathcal{F}(t,s))_{s \geq 0}$-stopping time for any $i \in \{1,\dots,n\}$ and $t \geq 0$. Then, the processes $(\tilde{U}^{\{j\}\{k\}}(t))_{t\geq 0}$ and $(\tilde{u}^{\{j\}\{k\}}(t))_{t\geq 0}$ are asymptotically equivalent, i.e.
	\begin{equation*}
		(\tilde{U}^{\{j\}\{k\}}(t) - \tilde{u}^{\{j\}\{k\}}(t)) \overset{\mathbb{P}}{\to} 0
	\end{equation*}
	for any $t \geq 0$.
\end{cor}
\begin{proof}
	Due to the optional stopping theorem and the independence of the observations of distinct patients, we also know that the stopped process $(\tilde{M}_i^{\{j\}\{k\}}(t,s))_{s \geq 0}$ is a martingale w.r.t. the extended filtration $(\mathcal{F}(t,s))_{s \geq 0}$ for any $i \in \{1,\dots,n\}$ and $t \geq 0$. The rest of the proof follows analogously to the proof of Theorem \ref{thm:transition_stat}.
\end{proof}

We do not only consider specific transitions, but also the entry into specific states or set of states, regardless of the origin of this transition. From now on, we need to distinguish whether we only want to consider the first hitting time of a state or a set of states or any hitting time of those and not only the first one. In the former case, we need to apply Corollary \ref{cor:transition_stat_stopped} while in the latter case Theorem \ref{thm:transition_stat} is sufficient. Hence, for a set of states $E \subset \{1,\dots,l\}$, we consider the processes $(U^{E}(t,s))_{t,s\geq 0}$ and $(U^{E}(t))_{t \geq 0}$ for the former resp. $(\Uall^{E}(t,s))_{t,s\geq 0}$ and $(\Uall^{E}(t))_{t \geq 0}$ for the latter case. Those are defined by
\begin{equation*}
	U^{E}(t,s)\coloneqq \frac{1}{\sqrt{n}} \sum_{k \in E} \sum_{j \notin E} \sum_{i=1}^n \int_{[0,s]} Q^{jk}(t,u) \left( Z_i - \frac{Y^{j\to E, Z=1}(t,u)}{Y^{j\to E}(t,u)} \right)\, N_i^{\{j\}\{k\}}(t, d(u \wedge T_i^E))
\end{equation*}
and
\begin{equation*}
	U^{E}(t)\coloneqq U^{E}(t,t).
\end{equation*}
where
\begin{equation*}
	Y^{j\to E}(t,s) = \sum_{i=1}^n Y^j_i(t,s) \cdot \mathbbm{1}_{\{s \leq T_i^E \wedge C_i(t)\}} \text{ and } Y^{j\to E, Z=1}(t,s) = \sum_{i=1}^n Y^j_i(t,s) \cdot \mathbbm{1}_{\{s \leq T_i^E \wedge C_i(t) \}} \cdot \mathbbm{1}_{\{Z_i=1\}},
\end{equation*}
resp.
\begin{equation*}
	\Uall^E(t,s)\coloneqq \sum_{j \notin E} \sum_{k \in E} U^{\{j\}\{k\}}(t,s) \quad \text{and} \quad \Uall^E(t)\coloneqq U^E(t,t)
\end{equation*}
In some cases, as e.g. in our example in Section \ref{sec:pfs_os} or in other purely progressive disease models, where composite events are given by a nested sequence of sets, the two approaches are the same as a set of states that has been hit once will never be left again.\\
As we will consider several events simultaneously, we need to establish multivariate convergence in probability using the following Lemma.
\begin{lemma}\label{lemma:mv_conv_in_prob}
	Let $(\mathbf{X}^{(n)})_{n \geq 0}$ be a sequence of $\mathbb{R}^d$-valued random vectors s.t. $X_c^{(n)} \overset{\mathbb{P}}{\to} X_c$ for each $c \in \{1,\dots, d\}$ as $n\to \infty$. Then, it also holds
	\begin{equation*}
		\mathbf{X}^{(n)} \overset{\mathbb{P}}{\to} \mathbf{X} \eqqcolon (X_1,\dots X_d)
	\end{equation*}
\end{lemma}
as $n \to \infty$ in $\mathbb{R}^d$.
\begin{proof}
	As all norms are equivalent on $\mathbf{R}^d$, it is enough to show it for the 1-norm, i.e.
	\begin{equation*}
		\mathbb{P}\left[\sum_{c=1}^d |X_c^{(n)} - X_c| > \varepsilon \right] \to 0
	\end{equation*}
	for any $\varepsilon > 0$. Because $\sum_{c=1}^d |X_c^{(n)} - X_c| > \varepsilon$ implies that there is at least one $c$ s.t. $|X_c^{(n)} - X_c| > \varepsilon/d$, we get
	\begin{align*}
		&\mathbb{P}\left[\sum_{c=1}^d |X_c^{(n)} - X_c| > \varepsilon \right]\\
		\leq & \mathbb{P}\left[\cup_{c=1}^d |X_c^{(n)} - X_c| > \varepsilon/d \right]\\
		\leq & \sum_{c=1}^d \mathbb{P}[|X_c^{(n)} - X_c| > \varepsilon/d]
	\end{align*}
	As all of the summands in the last sum converge to $0$, the sum becomes arbitrary small for increasing $n$.
\end{proof}
Using this Lemma, we can now establish asymptotic equivalence for multivariate processes which components are given as defined above.
\begin{cor}\label{cor:asymp_equiv_comp}
	Let $E \subset \{1,\dots,l\}$ be a subspace of the state space of a Markovian multi-state model. We assume that for any $k \in E$ and $j \notin E$ with $\lambda^{jk} \not\equiv 0$, the assumptions given in Theorem \ref{thm:transition_stat} hold and the assumption \ref{item:assumption:at_risk} holds as well for the adjusted "at risk"-processes $Y^{j \to E}$ and $y^{j \to E}$, respectively.\\
	Then, the processes $(U^E(t))_{t\geq 0}$ and $(u^E(t))_{t\geq 0}$ resp. $(\Uall^E(t))_{t\geq 0}$ and $(\Uall^E(t))_{t\geq 0}$ are asymptotically equivalent, i.e.
	\begin{equation*}
		(U^{E}(t) - u^{E}(t)) \overset{\mathbb{P}}{\to} 0 \quad \text{resp.} \quad (\Uall^{E}(t) - \uall^{E}(t)) \overset{\mathbb{P}}{\to} 0
	\end{equation*}
	for any $t \geq 0$ where
	\begin{equation}\label{eq:asymptotic_composite_processs}
		u^{E}(t,s)\coloneqq \frac{1}{\sqrt{n}} \sum_{k \in E} \sum_{j \notin E} \sum_{i=1}^n \int_{[0,s]} q^{jk}(t,u) \left( Z_i - \frac{y^{j\to E, Z=1}(t,u)}{y^{j\to E}(t,u)} \right)\, M_i^{\{j\}\{k\}}(t, d(u \wedge T_i^E))
	\end{equation}
	resp.
	\begin{equation}\label{eq:asymptotic_composite_processs_all}
	\uall^{E}(t,s)\coloneqq \sum_{k \in E} \sum_{j \notin E} u^{\{j\}\{k\}}(t,s)
	\end{equation}
	and $u^E(t)\coloneqq u^E(t,t)$ resp. $\uall^E(t)\coloneqq \uall^E(t,t)$.\\
	Moreover, for a set of subsets of the state space $E_1, \dots, E_d \subset \{1,\dots,l\}$, it holds
	\begin{equation}\label{eq:def_mv_stopped_process}
		(U^{E_1}(t),\dots,U^{E_d}(t))\eqqcolon \mathbf{U}(t) \overset{\mathbb{P}}{\to} \mathbf{u}(t) \coloneqq (u^{E_1}(t),\dots,u^{E_d}(t))
	\end{equation}
	and
	\begin{equation}\label{eq:def_mv_process}
		(\Uall^{E_1}(t),\dots,\Uall^{E_d}(t))\eqqcolon \mathbf{\Uall}(t) \overset{\mathbb{P}}{\to} \mathbf{\uall}(t) \coloneqq (\uall^{E_1}(t),\dots,\uall^{E_d}(t))
	\end{equation}
	for any $t\geq 0$ as $n \to \infty$.
\end{cor}
\begin{proof}
As convergence in probability holds for sums of random variables converging in probability, the statements concerning the convergence of a single component are a direct consequence of Corollary \ref{cor:transition_stat_stopped} resp. Theorem \ref{thm:transition_stat}.\\
As multivariate convergence in probability of $\mathbb{R}^d$-valued random variables can be established by convergence in probability of its $d$ components, the second part follows from Lemma \ref{lemma:mv_conv_in_prob}.
\end{proof}
Following Corollary \ref{cor:transition_stat_stopped} we will now consider the asymptotic distribution of the process $(\mathbf{u}(t))_{t \geq 0}$ and $(\mathbf{\uall}(t))_{t \geq 0}$ for the fixed set of (composite) events $E_1, \dots, E_d \subset \{1,\dots,l\}$. At first, we note, that the process $(u^{E_c}(t,s))_{s \geq 0}$ and $(\uall^{E_c}(t,s))_{s \geq 0}$ are $(\mathcal{F}(t,s))_{s\geq 0}$-martingales for any $c \in \{1,\dots,d\}$ as they are predictable processes integrated w.r.t. an $(\mathcal{F}(t,s))_{s\geq 0}$-martingale. Hence, expectation and covariance matrix of $\mathbf{u}(t)$ will be determined by 
\begin{equation*}
	\mathbb{E}\left[u^{E_c}(t)\right]=\mathbb{E}\left[\uall^{E_c}(t)\right]=0
\end{equation*}
and
\begin{align*}
	&\text{Cov}\left[u^{E_{c_1}}(t), u^{E_{c_2}}(t)\right]\\
	=&\sum_{k \in E_{c_1} \cap E_{c_2}} \sum_{j \notin E_{c_1} \cup E_{c_2}} \int_{[0,t]} q^{jk}(t,s)^2 \lambda^{jk}(s) \cdot \\
	& \qquad \cdot \mathbb{E}\left[ \mathbbm{1}_{\{X(s-)=j\}} \mathbbm{1}_{\{s \leq C(t) \wedge T^{E_{c_1} \cap E_{c_2}}\}} \left(Z - \frac{y^{j \to E_{c_1}, Z=1}(t,s)}{y^{j \to E_{c_1}}(t,s)}\right)\left(Z - \frac{y^{j \to E_{c_2}, Z=1}(t,s)}{y^{j \to E_{c_2}}(t,s)}\right) \right] ds.
\end{align*}
resp.
\begin{align*}
	&\text{Cov}\left[\uall^{E_{c_1}}(t), \uall^{E_{c_2}}(t)\right]\\
	=&\sum_{k \in E_{c_1} \cap E_{c_2}} \sum_{j \notin E_{c_1} \cup E_{c_2}} \int_{[0,t]} q^{jk}(t,s)^2 \lambda^{jk}(s) \mathbb{E}\left[ \mathbbm{1}_{\{X(s-)=j\}} \mathbbm{1}_{\{s \leq C(t)\}} \left(Z - \frac{y^{j, Z=1}(t,s)}{y^{j}(t,s)}\right)^2\right] ds.
\end{align*}
However, in order to apply Theorem II.5.1 of \cite{Andersen:93}, we want to show that $(\mathbf{u}(t))_{t \geq 0}$ and $(\mathbf{\uall}(t))_{t \geq 0}$ are martingales w.r.t. $(\mathcal{F}(t))_{t \geq 0}$ with $\mathcal{F}(t)\coloneqq \mathcal{F}(t,t)$. This is the filtration containing all available information at the respective calendar time $t$.
\begin{lemma}\label{lemma:mv_martingale}
Let $(\mathbf{u}(t))_{t\geq 0}$ and $(\mathbf{\uall}(t))_{t\geq 0}$ be the processes as defined in \eqref{eq:asymptotic_composite_processs}. Under the assumption that the limiting functions $q^{jk}(t,s)$, $\mu^{j\to E}(t,s)\coloneqq y^{j \to E, Z=1}(t,s)/y^{j \to E}(t,s)$ resp. $\mu^{j}(t,s)\coloneqq y^{j, Z=1}(t,s)/y^{j}(t,s)$ are independent of $t$ for any $j \in \{0,\dots,l\}$, $k \in \{1,\dots,l\}$ and $E \subset \{1,\dots,l\},$ both processes are $(\mathcal{F}(t))_{t \geq 0}$-martingales.
\end{lemma}
\begin{proof}
We will prove the statement for $(\mathbf{u}(t))_{t\geq 0}$. The proof for $(\mathbf{\uall}(t))_{t\geq 0}$ follows analogously. It is sufficient to consider the $d$ entries of $\mathbf{u}$ separately. For one fixed $E\subset \{1,\dots,l\}$ we decompose the process by
\begin{equation*}
	u^{E}(t)\coloneqq \frac{1}{\sqrt{n}} \sum_{k \in E} \sum_{j \notin E} \sum_{i=1}^n \underbrace{ \int_{[0,s]} q^{jk}(u) \left( Z_i - \mu^{j\to E}(u) \right)\, M_i^{\{j\}\{k\}}(t, d(u \wedge T_i^E))}_{\eqqcolon u_i^{jk,E}(t,s)}.
\end{equation*}
The process $(u_i^{jk,E}(t,s))_{s \geq 0}$ is an $(\mathcal{F}_i(t,s))_{s \geq 0}$-martingale where $\mathcal{F}_i(t,s) \coloneqq \sigma(\cup_{j,k} \mathcal{F}_i^{\{j\}\{k\}}(t,s)).$ As $u_i^{jk,E}(t,s) = u_i^{jk,E}(s+Y_i, (t-Y_i)_+)$ and $\mathcal{F}_i^{\{j\}\{k\}}(t,s) = \mathcal{F}_i^{\{j\}\{k\}}(s+Y_i,(t-Y_i)_+)$ for any $s,t\geq 0$ we know that for any $s\geq 0$ and $0\leq t_1 < t_2$ we have
\begin{align*}
	&\mathbb{E}\left[ u^E(t_2,s) |\mathcal{F}(t_1,s) \right]\\
	=&\frac{1}{\sqrt{n}} \sum_{i=1}^n \sum_{k \in E} \sum_{j \notin E \colon \lambda^{jk} \not\equiv 0} \mathbb{E}[u_i^{jk,E}(t_2,s) | \mathcal{F}(t_1,s)]\\
	=&\frac{1}{\sqrt{n}} \sum_{i=1}^n \sum_{k \in E} \sum_{j \notin E \colon \lambda^{jk} \not\equiv 0} \mathbb{E}[u_i^{jk,E}(t_2,s) | \mathcal{F}_i(t_1,s)]\\
	=&\frac{1}{\sqrt{n}} \sum_{i=1}^n \sum_{k \in E} \sum_{j \notin E \colon \lambda^{jk} \not\equiv 0} \mathbb{E}[u_i^{jk,E}(s+Y_i, (t_2-Y_i)_+) | \mathcal{F}_i(s + Y_i, (t_1 - Y_i)_+)]\\
	=&\frac{1}{\sqrt{n}} \sum_{i=1}^n \sum_{k \in E} \sum_{j \notin E \colon \lambda^{jk} \not\equiv 0} u_i^{jk,E}(s+Y_i, (t_1-Y_i)_+)\\
	=&u^E(t_1,s)
\end{align*}
Consequently, we also get
\begin{align*}
	\mathbb{E}[u^E(t_2)|\mathcal{F}(t_1, t_1)]&=\mathbb{E}[u^E(t_2,t_2)|\mathcal{F}(t_1, t_2)]=\mathbb{E}[\mathbb{E}[u^E(t_2,t_2)|\mathcal{F}(t_1, t_2)]|\mathcal{F}(t_1, t_1)]\\
	&=\mathbb{E}[u^E(t_2,t_1)|\mathcal{F}(t_1, t_1)]=u^E(t_1,t_1)=u^E(t_1)
\end{align*}
for any $0\leq t_1 < t_2$ by the tower property of conditional expectation.
\end{proof}
To this martingale we can now apply the multivariate version of Rebolledo's Theorem.
\begin{thm}\label{thm:rebolledo_application}
	The process $(\mathbf{u}(t))_{t \geq 0}$ converges as $n \to \infty$ in distribution to a Gaussian mean-zero vector martingale on the interval $[0, t_{\text{max}}]$ with the $d \times d$-matrix-valued covariance funtion $\mathbf{V}=(V_{c_1 c_2})_{c_1, c_2 \in \{1,\dots,d\}}$ given by
	\begin{equation}\label{eq:covariance_u}
		V_{c_1 c_2}(t)=\sum_{k \in E_{c_1} \cap E_{c_2}} \sum_{j \notin E_{c_1} \cup E_{c_2}} V^{jk}_{c_1 c_2}(t)
	\end{equation}
	where
	\begin{equation}\label{eq:covariance_contribution_u}
		V^{jk}_{c_1 c_2}(t) \coloneqq \int_{[0,t]} q^{jk}(s)^2 \lambda^{jk}(s) \mathbb{E}\left[\mathbbm{1}_{\{X(s)=j\}} \mathbbm{1}_{\{s \leq C(t) \wedge T^{E_{c_1} \cup E_{c_2}}\}} (Z - \mu^{j \to E_{c_1}}(s))(Z - \mu^{j \to E_{c_2}}(s)) \right]ds
	\end{equation}
	for $t \geq 0$ and $c_1, c_2 \in \{1, \dots, d\}$ if $q^{jk}$ is bounded for any $j \in \{0,\dots,l\}$ and $k \in \{1,\dots,l\}$ on that interval. This covariance function can consistently be estimated by $\hat{\mathbf{V}}=(\hat{V}_{c_1 c_2})_{c_1, c_2 \in \{1,\dots,d\}}$ given by
	\begin{equation*}
		\hat{V}_{c_1 c_2}(t)=\sum_{k \in E_{c_1} \cap E_{c_2}} \sum_{j \notin E_{c_1} \cup E_{c_2}} \frac{1}{n} \sum_{i=1}^n \hat{V}^{jk}_{c_1 c_2,i}(t)
	\end{equation*}
	where
	\begin{equation}\label{eq:covariance_component_estimator}
		\begin{split}
		&\hat{V}^{jk}_{c_1 c_2,i}(t) \coloneqq \int_{[0,t]} Q^{jk}(t,s)^2 \mathbbm{1}_{\{X_i(s)=j\}} \mathbbm{1}_{\{s \leq C_i(t) \wedge T^{E_{c_1} \cup E_{c_2}}\}} \cdot \\
		&\qquad \qquad \qquad \qquad \left(Z_i-\frac{Y^{j \to E_{c_1} , Z=1 } (t,s) }{Y^{j \to E_{c_1}} (t,s)}\right) \left(Z_i-\frac{Y^{j \to E_{c_2}, Z=1 } (t,s) }{Y^{j \to E_{c_2}} (t,s)}\right) \hat{\Lambda}^{jk, E_{c_1}\cup E_{c_2}}(t,ds).
		\end{split}
	\end{equation}
	In this formula, $\hat{\Lambda}^{jk, E_{c_1}\cup E_{c_2}}$ is the Nelson-Aalen estimator for the transition intensity of the transition $j\to k$ where all observations are censored at $T^{E_{c_1} \cup E_{c_2}}$.\\ 
	The same holds for the process $(\mathbf{u}_{\text{all}}(t))_{t \geq 0}$ where $\mathbf{V}$ and $\hat{\mathbf{V}}$ are replaced by $\mathbf{V}_{\text{all}}$ and $\hat{\mathbf{V}}_{\text{all}}$, respectively. The entries of these matrix-valued functions are given by
	\begin{equation}\label{eq:covariance_uall}
		V_{\text{all},c_1 c_2}(t)=\sum_{k \in E_{c_1} \cap E_{c_2}} \sum_{j \notin E_{c_1} \cup E_{c_2}} V^{jk}(t) \quad\text{resp.} \quad \hat{V}_{\text{all},c_1 c_2}(t)=\sum_{k \in E_{c_1} \cap E_{c_2}} \sum_{j \notin E_{c_1} \cup E_{c_2}} \frac{1}{n} \sum_{i=1}^n \hat{V}_i^{jk}(t) 
	\end{equation}
	with
	\begin{equation}\label{eq:covariance_contribution_uall}
		V^{jk}(t) \coloneqq \int_{[0,t]} q^{jk}(s)^2 \lambda^{jk}(s) \mathbb{E}\left[\mathbbm{1}_{\{X(s)=j\}} \mathbbm{1}_{\{s \leq C(t)\}} (Z - \mu^{j}(s))^2 \right]ds
	\end{equation}
	and 
	\begin{equation*}
		\hat{V}_i^{jk}(t) \coloneqq \int_{[0,t]} Q^{jk}(t,s)^2 \mathbbm{1}_{\{X_i(s)=j\}} \mathbbm{1}_{\{s \leq C_i(t)\}} \left(Z_i-\frac{Y^{j, Z=1 } (t,s) }{Y^{j} (t,s)}\right)^2 \hat{\Lambda}^{jk}(t,ds),
	\end{equation*}
	respectively.
\end{thm}
\begin{proof}
	We will prove this here for $(\mathbf{u}(t))_{t\geq 0}$. The proof for $(\mathbf{\uall}(t))_{t\geq 0}$ and corresponding variance components follows analogously.
	By splitting up the sums of which the components of $\mathbf{u}$ consist, exploiting the independence of observations concerning different transitions and applying standard rules to compute covariation processes, we obtain that the covariation process of $\mathbf{u}$ is given by
	\begin{equation*}
		\begin{split}
			\langle \mathbf{u} \rangle_{c_1 c_2} (t) &=\langle u^{E_{c_1}}, u^{E_{c_2}} \rangle (t)\\
			&=\sum_{k \in E_{c_1} \cap E_{c_2}} \sum_{j \notin E_{c_1} \cup E_{c_2}} \int_{[0,t]} q^{jk}(s)^2 \lambda^{jk}(s) \frac{1}{n} \sum_{i=1}^n \phi_i^{j,E_{c_1}}(t,s) \cdot \phi_i^{j,E_{c_2}}(t,s) ds
		\end{split}
	\end{equation*}
	where
	\begin{equation*}
	\phi_i^{j,E}(t,s) \coloneqq \left( \mathbbm{1}_{\{X_i(s)=j\}} \mathbbm{1}_{\{s \leq C_i(t) \wedge T_i^{E}\}} (Z_i - \mu^{j\to E}(s)) \right)
	\end{equation*}
	By a uniform law of large numbers and the boundedness of $\lambda^{jk}$ and $q^{jk}$ on bounded intervals for any $j$ and $k$ we obtain
	\begin{equation*}
		\langle \mathbf{u} \rangle_{c_1 c_2} (t) \overset{\mathbb{P}}{\to} \sum_{k \in E_{c_1} \cap E_{c_2}} \sum_{j \notin E_{c_1} \cup E_{c_2}} \int_{[0,t]} q^{jk}(s)^2 \lambda^{jk}(s) \mathbb{E}\left[\phi^{j,E_{c_1}}(t,s) \cdot \phi^{j,E_{c_2}}(t,s)\right] ds.
	\end{equation*}
	As there are $\mathbb{P}$-almost surely no simultaneous jumps of the processes considered here and by the boundedness of $q^{jk}$, there are no jumps of $u$ exceeding the size
	\begin{equation*}
		\frac{1}{\sqrt{n}} \sup_{j\in \{0,\dots,l\}, k \in \{1,\dots,l\}, u \in [0,t_{\text{max}}]} q^{jk}(u)
	\end{equation*}
	which converges to $0$ as $n\to \infty$, all conditions of Theorem II.5.1 of \cite{Andersen:93} are fulfilled. This yields the first part of the statement.\\
	For the second part of the statement, we note that the difference
	\begin{equation*}
		V^{jk}_{c_1 c_2}(t) - \hat{V}^{jk}_{c_1 c_2,i}(t)
	\end{equation*}
	can be decomposed into three parts in which only one factor differs. Each of these summands converges by the assumptions of Theorem \ref{thm:transition_stat}, a uniform law of large numbers for the sum
	\begin{equation*}
		\frac{1}{n} \sum_{i=1}^n \mathbbm{1}_{\{X_i(s)=j\}} \mathbbm{1}_{\{s \leq C_i(t) \wedge T_i^{E_{c_1} \cup E_{c_2}}\}} \left(Z_i-\frac{Y^{j \to E_{c_1} , Z=1 } (t,s) }{Y^{j \to E_{c_1}} (t,s)}\right) \left(Z_i-\frac{Y^{j \to E_{c_2}, Z=1 } (t,s) }{Y^{j \to E_{c_2}} (t,s)}\right)
	\end{equation*}
	and another application of Theorem II.5.1 of \cite{Andersen:93} by applying the same techniques as in the proof of Theorem \ref{thm:transition_stat}. As this holds for any $j$ and $k$, the second part of the statement follows.
\end{proof}
When estimating the variance of a single component, i.e. in the case $E_{c_1} = E_{c_2} = E$, the estimator from \eqref{eq:covariance_component_estimator} can be simplified to
\begin{equation*}
	\hat{V}^{jk}_{c_1 c_2,i}(t) = \int_{[0,t]} Q^{jk}(t,s)^2 \frac{Y^{j \to E, Z=1 } (t,s) }{Y^{j \to E} (t,s)}\cdot \left( 1 - \frac{Y^{j \to E, Z=1 } (t,s) }{Y^{j \to E} (t,s)} \right) N_i^{\{j\}\{k\}}(t,d(s \wedge T_i^{E})).
\end{equation*}
Analogous simplifications can be made for covariance components if $T^{E_{c_1}} \leq T^{E_{c_2}}$ or $E_{c_1} \cup E_{c_2} = \{1, \dots l\}$. This is the case for both examples given in the main manuscript.
\begin{cor}\label{cor:test_stat}
	For a sequence of analysis dates in calendar time $0\eqqcolon < t_1 < \dots < t_m$ of analysis dates, the following holds:
	\begin{enumerate}[label = (C\arabic*)]
		\item The test statistics $\mathbf{U}$ at the calendar dates $t_1,\dots,t_m$ are asymptotically jointly normally distributed and hence also have asymptotically independent increments, i.e. 
		\begin{align*}
			&(\mathbf{U}(t_1),\dots,\mathbf{U}_M(t_m)) &\underset{n \to \infty}{\overset{\mathcal{D}}{\to}} \mathcal{N}(0, \mathbf{V}_{\text{acc}})\\
			&(\mathbf{U}(t_1) - \mathbf{U}_M(t_0),\dots,\mathbf{U}_M(t_m)-\mathbf{U}_M(t_{p-1})) &\underset{n \to \infty}{\overset{\mathcal{D}}{\to}} \mathcal{N}(0, \mathbf{V}_{\text{inc}})
		\end{align*}
		where both $\mathbf{V}_{\text{acc}}$ and $\mathbf{V}_{\text{inc}}$ are $md \times md$ matrices consisting of $m^2$ blocks of size $d\times d$. The block in row $r_1$ and column $r_2$ of $\mathbf{V}_{\text{acc}}$ is given by $\mathbf{V}(t_{r_1} \wedge t_{r_2})$ and $\mathbf{V}_{\text{inc}}$ is a block diagonal matrix with $\mathbf{V}_{\text{inc}}=\text{diag}(\mathbf{V}(t_1) - \mathbf{V}(t_0), \dots, \mathbf{V}(t_m) - \mathbf{V}(t_{m-1}))$.
		\item \label{item:test_stat} If $\mathbf{V}(t_{r}) - \mathbf{V}(t_{r-1})$ is invertible for any $r \in \{1,\dots,m\}$, it holds
		\begin{equation*}
			S_r \coloneqq (\mathbf{U}(t_r) - \mathbf{U}(t_{r-1}))^T (\hat{\mathbf{V}}(t_{r}) - \hat{\mathbf{V}}(t_{r-1}))^+ (\mathbf{U}(t_r) - \mathbf{U}(t_{r-1})) \underset{n \to \infty}{\overset{\mathcal{D}}{\to}} \chi^2_{d} \qquad \forall r \in \{1,\dots,m\}
		\end{equation*}
		where $\mathbf{V}^+$ denotes the generalized Moore-Penrose inverse for some matrix $\mathbf{V}$. Also, $S_1,\dots, S_m$ are mutually asymptotically independent.
	\end{enumerate}
	Both statements also hold when replacing $\mathbf{U}$, $\mathbf{V}$ and $\hat{\mathbf{V}}$ by $\mathbf{\Uall}$, $\mathbf{\Vall}$ and $\hat{\mathbf{V}}_{\text{all}}$.
\end{cor}
\begin{proof}
	The first statement holds for $\mathbf{u}$ by Theorem \ref{thm:transition_stat}. The components of $\mathbf{U}$ and $\mathbf{u}$ are asymptotically equivalent for each $t\geq 0$ by Theorem \ref{thm:transition_stat}. By Lemma \ref{lemma:mv_conv_in_prob}, this asymptotical equivalence also holds for the whole vectors $\mathbf{U}$ and $\mathbf{u}$. If for two sequences of random variables $(X_n)_{n\geq 0}$ and $(Y_n)_{n\geq 0}$, it holds $X_n \overset{\mathcal{D}}{\to} X$ for a random variable $X$ and $X_n-Y_n \overset{\mathbb{P}}{\to} 0$, then $Y_n$ also converges in distribution to $X$. Hence, the first statment also holds for $\mathbf{U}$.\\
	For the second statement, we note that $\hat{\mathbf{V}}(t_{r}) - \hat{\mathbf{V}}(t_{r-1})$ converges in probability to $\mathbf{V}(t_{r}) - \mathbf{V}(t_{r-1})$. As we assume that the latter is invertible, we also have
	\begin{equation*}
		(\hat{\mathbf{V}}(t_{r}) - \hat{\mathbf{V}}(t_{r-1}))^+ \overset{\mathbb{P}}{\to} (\mathbf{V}(t_{r}) - \mathbf{V}(t_{r-1}))^+ = (\mathbf{V}(t_{r}) - \mathbf{V}(t_{r-1}))^{-1}
	\end{equation*}	
	as stated e.g. in \cite{Puri:1984}. By Slutsky's Theorem and Lemma \ref{lemma:mv_conv_in_prob} it follows
	\begin{equation*}
		(S_1,\dots, S_m) \overset{\mathbb{P}}{\to} (\mathbf{u}(t_1)^T \mathbf{V}(t_{1})^{-1} \mathbf{u}(t_1), \dots, (\mathbf{u}(t_{m}) - \mathbf{u}(t_{m-1}))^T (\mathbf{V}(t_{m}) - \mathbf{V}(t_{m-1}))^{-1} (\mathbf{u}(t_{m}) - \mathbf{u}(t_{m-1}))).
	\end{equation*}
	By the first part of the statement, the continuous mapping theorem and the same argument as in the last step of the proof of the firstpart, we obtain
	\begin{equation*}
		(S_1,\dots, S_m) \overset{\mathcal{D}}{\to} (X_1, \dots, X_m)
	\end{equation*}
	where $X_1,\dots, X_m$ are $m$ independent $\chi^2_d$-distributed random variables.\\
	The proof for $\mathbf{\Uall}$, $\mathbf{\Vall}$ and $\hat{\mathbf{V}}_{\text{all}}$ follows analogously.
\end{proof}

Now, that the test statistics and its asymptotic distribution are fixed, we will also consider the distribution under alternatives. In doing so, we assume that the transition intensities $\lambda^{jk}_0$ and $\lambda^{jk}_1$ for the transition $j\to k$ in the treatment groups in which $Z=0$ resp. $Z=1$ only differ by a factor, i.e. we assume transition-wise proportional hazards. In particular, we will consider contiguous alternatives. Hence, $\lambda^{jk}_1 / \lambda^{jk}_0 = 1 - n^{-1/2}\tilde{\delta}^{jk}$ for some $\tilde{\delta}^{jk}$ independent of $n$ for each transition $j \to k$.\\
Under some alternative, we need to consider that the transition specific martingale $M_i^{\{j\}\{k\}}$ from \eqref{eq:transition_specific_martingale} is then given by
\begin{equation}\label{eq:martingale_under_alternative}
	M_i^{\{j\}\{k\}}(t,s) \coloneqq N_i^{\{j\}\{k\}}(t,s) - \int_{[0,s]} Y_i^{j, Z=0}(t,u) \lambda^{jk}_0(u) + Y_i^{j, Z=1}(t,u) \lambda^{jk}_1(u) \, du.
\end{equation}

\begin{thm}\label{thm:transition_stat_alternative}
Let $j \neq k$ be two states of a Markovian multi-state model with $\lambda^{jk} \not\equiv 0$. Let $\lambda_0^{jk} \coloneqq \lambda^{jk}$ denote the corresponding transition intensity function if $Z=0$. Under contiguous alternatives, the transition intensity if $Z=1$ is then given by $\lambda^{jk}_1 \coloneqq \lambda^{jk} \cdot (1 - n^{-1/2} \tilde{\delta}^{jk})$. Additionally, we define the drift process $(\theta^{\{j\}\{k\}}(t,s))_{t,s \geq 0}$ by
\begin{equation*}
	\theta^{\{j\}\{k\}}(t,s)\coloneqq -\tilde{\delta}^{jk} \int_{[0,s]} q^{jk}(t,u) \left(1 - \frac{y^{j,Z=1}(t,u)}{y^j(t,u)} \right) y^{j,Z=1}(t,u) \lambda^{jk}_0(u) \, du 
\end{equation*}
If the conditions \ref{item:assumption:weight_convergence} - \ref{item:assumption:weight_bound} are given under the contiguous alternatives, the processes $(U^{\{j\}\{k\}}(t))_{t\geq 0}$ and $(u^{\{j\}\{k\}}(t) + \theta^{\{j\}\{k\}}(t))_{t \geq 0}$ are asymptotically equivalent, i.e.
\begin{equation*}
	U^{\{j\}\{k\}}(t) - (u^{\{j\}\{k\}}(t) + \theta^{\{j\}\{k\}}(t)) \overset{\mathbb{P}}{\to} 0
\end{equation*}
for any $t \geq 0$.
\end{thm}

\begin{proof}
As in the proof of Theorem \ref{thm:transition_stat}, we obtain 
\begin{equation*}
	U^{\{j\}\{k\}}(t,s) = \frac{1}{\sqrt{n}} \sum_{i=1}^n \int_{[0,t]} Q^{jk}(t,u) \left( Z_i - \frac{Y^{j, Z=1}(t,u)}{Y^j(t,u)} \right) \left(N_i^{\{j\}\{k\}}(t, du) - \lambda^{jk}(u) \cdot Y^j_i(t,u)\, du \right)
\end{equation*}
In contrast to Theorem \ref{thm:transition_stat}, the integrator is not a maringale. However, we can decompose this process to obtain the martingale \eqref{eq:martingale_under_alternative} as
\begin{align*}
	U^{\{j\}\{k\}}(t)=&\frac{1}{\sqrt{n}} \sum_{i=1}^n \int_{[0,t]} Q^{jk}(t,u) \left( Z_i - \frac{Y^{j, Z=1}(t,u)}{Y^j(t,u)} \right) M_i^{\{j\}\{k\}}(t, du)\\
	&+\frac{1}{\sqrt{n}} \sum_{i=1}^n \int_{[0,t]} Q^{jk}(t,u) \left( Z_i - \frac{Y^{j, Z=1}(t,u)}{Y^j(t,u)} \right) \cdot \lambda_0^{jk}(u) \cdot n^{-1/2} \tilde{\delta}^{jk} \cdot Y_i^{j,Z=1}(t,u)\, du.
\end{align*}
This basically follows from the equality $\lambda_0^{jk}(u) = \lambda_1^{jk}(u) + \lambda_0^{jk}(u) \cdot n^{-1/2} \tilde{\delta}^{jk}$. For the first summand, the proof of Theorem \ref{thm:transition_stat} can be applied. For the second summand, we can use the assumptions given in this statement and the law of large number to obtain
\begin{equation*}
	\theta^{\{j\}\{k\}}(t,t) - \frac{1}{n} \sum_{i=1}^n \int_{[0,t]} Q^{jk}(t,u) \left(Z_i - \frac{Y^{j, Z=1}(t,u)}{Y^j(t,u)} \right) \cdot \lambda_0^{jk}(u) \cdot \tilde{\delta}^{jk} \cdot Y_i^{j,Z=1}(t,u)\, du \overset{\mathbb{P}}{\to} 0.
\end{equation*}
\end{proof}

Concerning the validity of assumption \ref{item:assumption:at_risk}, we note the following. The $(l+1)\times(l+1)$-matrix valued transition intensity functions $\mathbf{\Lambda}^0$ and $\mathbf{\Lambda}^1(s)$ for the two treatment groups are related to each other in the following way under the contiguous alternatives mentioned above:
\begin{equation*}
	\mathbf{\Lambda}^1(s) = \mathbf{\Lambda}^0(s) \circ (1 - n^{-1/2} \tilde{\delta}^{jk})_{j,k \in \{0,\dots,l\}}.
\end{equation*}
Here, $a \circ b$ denotes the Hadamard product (i.e. elementwise multiplication) of two matrices $a$ and $b$ of the same size. The transition probability matrix $\mathbf{P}_g(s_1,s_2)$ is given by the exponential of $\mathbf{\Lambda}^g(s_2)-\mathbf{\Lambda}^g(s_1)$ for any group $g\in \{0,1\}$. As the matrix exponential is Lipschitz continuous on compact subsets of the space of $(l+1)\times (l+1)$-matrices and we assume the integrated transition intensities to be bounded on compact subsets of $[0,\infty)$, the second condition should in general be fulfilled, i.e. we have $\mathbf{P}_1(\cdot,\cdot) \to \mathbf{P}_0(\cdot,\cdot)$ elementwise uniformly on compact subsets of $\mathbb{R}_+^2$.\\
Based on Theorem \ref{thm:transition_stat_alternative}, analogous results to those of Corollary \ref{cor:transition_stat_stopped}, Corollary \ref{cor:asymp_equiv_comp}, Lemma \ref{lemma:mv_martingale}, Theorem \ref{thm:rebolledo_application} and Corollary \ref{cor:test_stat} can be derived. The most important results from that procedure are summarized in the following Corollary. The key factor here is that the drift terms $\theta^{\{j\}\{k\}}$ are deterministic, and hence it does not contribute any variance. For the random part $u$ of the process obtained in Theorem \ref{thm:transition_stat_alternative}, analogous results as those obtained above under the null hypothesis apply. For any event defined by a subset of the state space $E$, we define the drift processes $(\theta^{E}(t))_{t\geq 0}$ and $(\theta^E_{\text{all}}(t))_{t\geq 0}$ by
\begin{equation*}
	\theta^{E}(t)\coloneqq \sum_{k \in E} \sum_{j \notin E\colon \lambda^{jk}\not\equiv 0} \underbrace{ -\tilde{\delta}^{jk} \int_{[0,t]} q^{jk}(u) \left(1 - \mu^{j \to E}(u) \right) y^{j\to E,Z=1}(t,u) \lambda^{jk}_0(u) \, du}_{\eqqcolon \theta^{jk,E}(t)}.
\end{equation*} 
resp. 
\begin{equation*}
	\theta^{E}_{\text{all}}(t)\coloneqq \sum_{k \in E} \sum_{j \notin E\colon \lambda^{jk}\not\equiv 0}  -\tilde{\delta}^{jk} \int_{[0,t]} q^{jk}(u) \left(1 - \mu^{j}(u) \right) y^{j,Z=1}(t,u) \lambda^{jk}_0(u) \, du.
\end{equation*}
For a collection of subsets $E_1,\dots,E_d$ of subsets of the state space, we define the $d$-dimensional processes $\boldsymbol{\theta}\coloneqq (\theta^{E_1}, \dots, \theta^{E_d})$ and $\boldsymbol{\theta}_{\text{all}}\coloneqq (\theta^{E_1}_{\text{all}}, \dots, \theta^{E_d}_{\text{all}})$.
\begin{cor}
	Under contiguous alternatives for each transition as in Theorem \ref{thm:transition_stat_alternative}, the following statements hold under assumptions \ref{item:assumption:weight_convergence} - \ref{item:assumption:weight_bound} and the assumption stated in Lemma \ref{lemma:mv_martingale}:
	\begin{enumerate}
		\item The processes $\mathbf{U}$ and $\mathbf{U}_{\text{all}}$ as defined in \eqref{eq:def_mv_stopped_process} resp. \eqref{eq:def_mv_process} are asymptotically equivalent to the processes $\mathbf{u} + \boldsymbol{\theta}$ resp. $\mathbf{u}_{\text{all}} + \boldsymbol{\theta}_{\text{all}}$. where $\mathbf{u}$ and $\mathbf{u}_{\text{all}}$ are $(\mathcal{F}(t))_{t \geq 0}$-martingales which converge in distribution to a Gaussian mean-zero vector martingale with covariance processes as stated in Theorem \ref{thm:rebolledo_application}.
		\item The increments of $\mathbf{U}$ resp. $\mathbf{U}_{\text{all}}$ are asymptotically independent and asymptotically multivariately normally distributed. The covariance can be estimated consistently as stated in Theorem \ref{thm:rebolledo_application}.
		\item The stage-wise test statistics $S_1, \dots, S_m$ as defined in Corollary \ref{cor:test_stat} are asymptotically independent and asymptotically follow a non-central $\chi^2_d$ distribution with non-centrality parameter
		\begin{equation*}
			\eta_r \coloneqq (\boldsymbol{\theta}(t_r) - \boldsymbol{\theta}(t_{r-1}))^T (\mathbf{V}(t_{r}) - \mathbf{V}(t_{r-1}))^{-1} (\boldsymbol{\theta}(t_r) - \boldsymbol{\theta}(t_{r-1}))
		\end{equation*}
		resp.
		\begin{equation*}
			\eta_{r, \text{all}} \coloneqq (\boldsymbol{\theta}_{\text{all}}(t_r) - \boldsymbol{\theta}_{\text{all}}(t_{r-1}))^T (\mathbf{V}_{\text{all}}(t_{r}) - \mathbf{V}_{\text{all}}(t_{r-1}))^{-1} (\boldsymbol{\theta}_{\text{all}}(t_r) - \boldsymbol{\theta}_{\text{all}}(t_{r-1}))
		\end{equation*}
		if the increments of the covariance function are invertible.
	\end{enumerate}
\end{cor}
The last part of this statement can now be used for power and sample size calculations. However, we notice strong deviations for small sample sizes if such results are used for sample size considerations. This is particularly the case if differences in several transitions are present. In such a case, it is better to approximate the distribution of an increment $\mathbf{U}(t_r) - \mathbf{U}(t_{r-1})$ by a normal distribution with adjusted drift components
\begin{equation*}
	\tilde{\theta}^{jk, E}(t)\coloneqq - \tilde{\delta}^{jk} \int_{[0,t]} q^{jk}(u) \left(1 - \frac{y^{j\to E,Z=1}(t,u)}{y^{j \to E}(t,u)} \right) y^{j\to E,Z=1}(t,u) \lambda^{jk}_0(u) du
\end{equation*}
and adjusted variance components 
\begin{equation*}
	\begin{split}
	\tilde{V}^{jk}_{c_1 c_2}(t) \coloneqq \int_{[0,t]} q^{jk}(s)^2 &(y^{j\to E_{c_1} \cup E_{c_2}, Z=0}(t,s)\lambda_0^{jk}(s) + y^{j\to E_{c_1} \cup E_{c_2}, Z=1}(t,s)\lambda_1^{jk}(s)) \cdot\\
	&\cdot \frac{y^{j \to E_{c_1}, Z=1}(t,s)}{y^{j \to E_{c_1}}(t,s)} \frac{y^{j \to E_{c_2}, Z=1}(t,s)}{y^{j \to E_{c_2}}(t,s)} ds
	\end{split}.
\end{equation*}
The quantities used here can be computed from the planning assumptions. Here, it should be noted, that in the main manuscript, we are considering local alternatives of the form $\lambda^{jk}_1 / \lambda^{jk}_0 = \delta^{jk}$ and that the quantities calculated there result from simple algebraic transformations of the calculations here,
\subsection*{A.2 Invertibility of the increments of the covariance matrix}\label{sec:invertibility}

For the test statistics obtained from \ref{item:test_stat} of Corollary \ref{cor:test_stat}, it was required that the increments of the covariance function are invertible. In what follows we want to provide criteria under which this invertibility can be guaranteed. While for $\mathbf{\Vall}$, a necessary and sufficient characterization can be provided, this is not possible for $\mathbf{V}$. The reason for this lies in the different structures of those matrices. For one specific transition $j \to k$, its contributions to different entries of the covariance function $\mathbf{\Vall}$ are all the same as it can be seen from \eqref{eq:covariance_uall}. Contrarily, this is not the case for $\mathbf{V}$ as it can be seen from \eqref{eq:covariance_u}. The contributions differ, depending on the combination of $E_{c_1}$ and $E_{c_2}$.\\
Throughout the next part, we will assume that $V^{jk}$ and $V_{c_1 c_2}^{jk}$ as defined in \eqref{eq:covariance_contribution_uall} and \eqref{eq:covariance_contribution_u}, respectively, are strictly increasing functions for any transition $j \to k$ included in our multi-state model and $c_1,c_2 \in \{1,\dots, d\}$, in particular this means that for any arbitrarily small interval $[t_1, t_2]$, there is a positive probability that a transition from state $j$ to state $k$ before occurs between the calendar time $t_1$ and $t_2$.\\
As mentioned above, it is easier to establish a criterion for the invertibility of increments of $\mathbf{\Vall}$. Let $w_1,\dots,w_{n_{\text{trans}}}$ be an arbitrary enumeration of the transitions in the multi-state model. For any subset $E \subset \{1,\dots,l\}$ of the state space, we define the vector $\Psi_E\coloneqq (\psi_{E,j})_{j \in \{1,\dots,n_{\text{trans}}\}}$ by
\begin{equation*}
	\psi_{E,j} \coloneqq 
		\begin{cases}
			1 & \text{if $w_j$ contributes to $E$}\\
			0 & \text{else.}
		\end{cases}
\end{equation*} 
where "contribution" of a transition to an event $E$ means that the start of the transition starts outside of $E$ and enters $E$.
\begin{lemma}\label{lemma:invertibility_Vall}
	Increments of the covariance matrix function $\mathbf{\Vall}$ are invertible if and only if	for the sets $E_1,\dots, E_d$, the family of vectors $(\Psi_{E_c})_{c \in \{1,\dots,d\}}$ is linearly independent.
\end{lemma}
\begin{proof}
	Let $0\leq t_1 < t_2$ be arbitrary. For any $j \in \{1,\dots,n_{\text{trans}}\}$, we define $\phi_j \coloneqq V^{jk}(t_2) - V^{jk}(t_1)$ if the $j$-th transition of our enumeration goes from state $j$ to state $k$. If $a \circ b$ denotes the Hadamard product (i.e. elementwise multiplication) and $a \cdot b$ the scalar product of two vectors $a$ and $b$, the increment of the covariance matrix $\mathbf{\Vall}$ is given by
	\begin{equation*}
		\mathbf{\Vall}(t_2) - \mathbf{\Vall}(t_1) = 
		\begin{pmatrix}
			(\Psi_1 \circ \Psi_1) \cdot \phi & \cdots & (\Psi_1 \circ \Psi_d) \cdot \phi\\
			\vdots & \ddots & \vdots\\
			(\Psi_d \circ \Psi_1) \cdot \phi & \cdots & (\Psi_d \circ \Psi_d) \cdot \phi
		\end{pmatrix}
	\end{equation*} 
	Additionally, we note that $\mathbf{\Vall}(t_2) - \mathbf{\Vall}(t_1)$ is a covariance matrix and is thus symmetric and positive semi-definite. Hence, its invertibility is equivalent to positive definiteness. In what follows, we will also abbreviate $\Psi_{E_c}\coloneqq \Psi_{c}$ and $\psi_{E_c,j}\coloneqq \psi_{c,j}$ for any $c$ and $j \in \{1,\dots,n_{\text{trans}}\}$.\\
	\underline{"$\Rightarrow$"}: We prove by contrapositive. If $(\Psi_{c})_{c \in \{1,\dots,d\}}$ is not linearly independent, there is some $\beta \in \mathbb{R}^d$ s.t. $\sum_{c=1}^d \beta_c \Psi_c=0$. Hence, for any $c\in \{1,\dots,d\}$ it holds
	\begin{equation*}
		\sum_{i=1}^d \beta_i ((\Psi_c \circ \Psi_i)\cdot v) = (\Psi_c \circ \underbrace{(\sum_{i=1}^d \beta_i\Psi_i)}_{=0})\cdot v = 0.
	\end{equation*}
	Followingly, $\mathbf{\Vall}(t_2) - \mathbf{\Vall}(t_1)$ does not have full rank and is thus not invertible.\\
	\underline{"$\Leftarrow$"}: We prove by contrapositive. If $\mathbf{\Vall}(t_2) - \mathbf{\Vall}(t_1)$ is not invertible, it is also not positive definite. Hence, there is some $0 \neq \beta \in \mathbb{R}^d$ s.t. $\beta^T (\mathbf{\Vall}(t_2) - \mathbf{\Vall}(t_1)) \beta = 0$. If $\sqrt{\phi}$ denotes the element-wise square root of the (element-wise positive) vector $\phi$, then
	\begin{align*}
		0&=\beta^T (\mathbf{\Vall}(t_2) - \mathbf{\Vall}(t_1)) \beta\\
		&=\sum_{i=1}^d \sum_{j=1}^d \beta_i \beta_j (\Psi_i \circ \Psi_j) \cdot \phi\\
		&=\sum_{i=1}^d \sum_{j=1}^d (\beta_i\Psi_i \circ \beta_j\Psi_j) \cdot \phi\\
		&=\sum_{i=1}^d \sum_{j=1}^d \underbrace{(\beta_i\Psi_i \circ \sqrt{\phi})\cdot(\beta_j\Psi_j \circ \sqrt{\phi})}_{= \langle (\beta_i\Psi_i \circ \sqrt{\phi}), (\beta_j\Psi_j \circ \sqrt{\phi}) \rangle}\\
		&=\langle \sum_{i=1}^d (\beta_i\Psi_i \circ \sqrt{\phi}), \sum_{j=1}^d (\beta_j\Psi_j \circ \sqrt{\phi}) \rangle\\
		&=||\sum_{i=1}^d (\beta_i\Psi_i \circ \sqrt{\phi})||^2. 
	\end{align*}
	Accordingly, each entry of $\sum_{i=1}^d (\beta_i\Psi_i \circ \sqrt{\phi})$ is equal to zero. Hence, for any $j \in \{1,\dots,n_{\text{trans}}\}$,
	\begin{align*}
		&\sum_{i=1}^d (\beta_i\psi_{i,j} \sqrt{\phi_j}) = 0\\
		\Rightarrow & \sqrt{\phi_j} \sum_{i=1}^d (\beta_i\psi_{i,j}) = 0\\
		\Rightarrow & \sum_{i=1}^d (\beta_i\psi_{i,j}) = 0.
	\end{align*}
	It follows that $(\Psi_{E_c})_{c \in \{1,\dots,d\}}$ is not linearly independent.
\end{proof}
As explained above, the situation for increments of $\mathbf{V}$ is a bit more complicated. In what follows, we try to give a criterion as general as possible.\\ 
Let $I_0\coloneqq \{1,\dots,d\}$ denote the index set of the events we want to assess. The corresponding variance matrix is given by $\mathbf{V}$. For any subset of this set, we denote by $\mathbf{V}_{I}$ the $|I|\times|I|$-submatrix of $\mathbf{V}$ in which all columns and rows in $I_0\setminus I$ are deleted.\\
We refer to a transition $j \to k$ as \textit{exclusive} for a set $E_c$ w.r.t. some collection $I \subset I_0$ of events if $j \notin E_c$, $k \in E_c$ and $j \in E_{\tilde{c}}$ or  $k \notin E_{\tilde{c}}$ for any $\tilde{c} \in I \setminus \{c\}$.
Given these definitions we now state the following
\begin{lemma}
For some $I\subseteq I_0$, the increments of the covariance matrix function $\mathbf{V}_I$ are invertible if at least one of the following conditions is fulfilled:
\begin{enumerate}[label = (I\arabic*)]
	\item\label{item:exclusivity} For each $c\in I$, $E_c$ has an exclusive transition w.r.t. $I$.
	\item\label{item:decomposition} The index sets $I$ can be decomposed into disjoint sets $I_1 \cup I_2$ s.t. for each $c\in I_1$, $E_c$ has an exclusive transition w.r.t. $I$ and the increments of $\mathbf{V}_{I_2}$ are invertible.
\end{enumerate}
\end{lemma}
\begin{proof}
	\underline{\ref{item:exclusivity}:} Let $j_c \to k_c$ denote the transition that is exclusive for $c \in I$ w.r.t. $I$. We define
	\begin{equation*}
		\tilde{u}^{E_c}(t)\coloneqq u^{E_c}(t) - \frac{1}{\sqrt{n}} \sum_{i=1}^n \int_{[0,t]} q^{j_c k_c}(t,u) \left( Z_i - \frac{y^{j_c\to E, Z=1}(t,u)}{y^{j_c\to E}(t,u)} \right)\, M_i^{\{j_c\}\{k_c\}}(t, d(u \wedge T_i^E)).
	\end{equation*}
	The covariance matrix of the random vector $\tilde{\mathbf{u}}_I \coloneqq (\tilde{u}^{E_c})_{c \in I}$ is given by $\tilde{\mathbf{V}}_I$ where the off-diagonal entries are the same as those of $\mathbf{V}_I$ and the diagonal entry belonging to some $E_c$ with $c \in I$ is given by $V_{cc} - V_{cc}^{j_c k_c}$. As $\tilde{\mathbf{V}}_I$ is a matrix-valued covariance function, its increments are positive semi-definite. The $|I|\times|I|$-diagonal matrix with entries $V_{cc}^{j_c k_c}$ for $c \in I$ has positive definite increments by the assumption stated at the beginning of this subsection. Hence, the increments of $\mathbf{V}_I$ can be written as a sum of a positive semi-definite and a positive definite matrix. They are hence positive definite and invertible.\\
	\underline{\ref{item:decomposition}:} For each $c \in I_1$, we denote by $j_c \to k_c$ the transition that is exclusive for $E_c$ w.r.t. $I$. We define
	\begin{equation*}
		\tilde{u}^{E_c}(t)\coloneqq u^{E_c}(t) - \mathbbm{1}_{\{c \in I_1\}}\cdot \frac{1}{\sqrt{n}} \sum_{i=1}^n \int_{[0,t]} q^{j_c k_c}(t,u) \left( Z_i - \frac{y^{j_c\to E, Z=1}(t,u)}{y^{j_c\to E}(t,u)} \right)\, M_i^{\{j_c\}\{k_c\}}(t, d(u \wedge T_i^E)).
	\end{equation*}
	The covariance matrix of the random vector $\tilde{\mathbf{u}}_I \coloneqq (\tilde{u}^{E_c})_{c \in I}$ is given by $\tilde{\mathbf{V}}_I$ where the off-diagonal entries are the same as those of $\mathbf{V}_I$ and the diagonal entry belonging to some $E_c$ with $c \in I$ is given by $V_{cc} - \mathbbm{1}_{\{c \in I_1\}}\cdot V_{cc}^{j_c k_c}$. As $\tilde{\mathbf{V}}_I$ is a matrix-valued covariance function, its increments are positive semi-definite. The $|I|\times|I|$-diagonal matrix with non-negative entries $\mathbbm{1}_{\{c \in I_1\}} V_{cc}^{j_c k_c}$ for $c \in I$ also has positive semi-definite increments by the assumption stated before Lemma \ref{lemma:invertibility_Vall}. Now, if $\mathbf{V}_I$ is not invertible, it is also not positive definite. Hence there must be some $0\neq\beta\in \mathbb{R}^{|I|}$ s.t.
	\begin{align*}
		0&=\beta^T \mathbf{V}_I \beta\\
		&=\underbrace{\beta^T \tilde{\mathbf{V}}_I \beta}_{\geq 0} + \underbrace{\beta^T \text{diag}((\mathbbm{1}_{\{c \in I_1\}} V_{cc}^{j_c k_c})_{c \in I}) \beta}_{=0 \text{ only if } \beta_c=0 \; \forall c \in I_1}.
	\end{align*}
	Hence, this equation can only be fulfilled, if $\beta_c=0$ for all $c \in I$. For such a $\beta$ we can thus rewrite
	\begin{equation*}
		\beta^T \mathbf{V}_I \beta = \sum_{c_1 \in I_2} \sum_{c_2 \in I_2} \beta_{c_1} \beta_{c_2} V_{c_1 c_2} = \beta_{I_2}^T \mathbf{V}_{I_2} \beta_{I_2}
	\end{equation*}
	where $\beta_{I_2}$ is the subvector of $\beta$ consisting only of the components belonging to $I_2$. If $\mathbf{V}_{I_2}$ has positive definite increments, this is not possible and $\mathbf{V}_{I}$ must be invertible.
\end{proof}
This result already ensures invertibility of increments of the covariance matrix for the settings depicted by Figures \ref{fig:pfsos} and \ref{fig:efficacy_safety_msm}. Going beyond that, it also holds in progressive illness models as in \cite{Wanek:1993} in which the events are defined by decreasing sequence of sets of nodes.\\
This Lemma can also be applied iteratively if the condition \ref{item:decomposition} is fulfilled and it remains to show invertibility of a submatrix. The next statement describes how the invertibility of the required in the second part of \ref{item:decomposition} can also be verified. For this statement, we define by $\text{In}(E)$ the set of transitions of the multi-state model going from outside of $E$ in to $E$.
\begin{lemma}\label{lemma:irreducibility}
	The increments of the covariance matrix function $\mathbf{V}_I$ are invertible if the follwoing conditions hold:
	\begin{enumerate}[label = (D\arabic*)]
		\item\label{item:irreducibility} For any $c_0, c\in I$, there is a finite sequence $c_0, c_1, \dots, c_n \coloneqq c$ s.t. $\text{In}(E_{c_{i-1}}) \cap \text{In}(E_{c_i}) \neq \emptyset$ for all $i \in \{1,\dots,n\}$.
		\item\label{item:diagonal_dominance} For any ${c_1}\in I$ and any transition $j\to k$ in $\text{In}(E_{c_1})$, there is at most one $c_1\neq c_2 \in I$ s.t. the transition is also in $\text{In}(E_{c_2})$.
		\item\label{item:strict dominance} For at least one $c_1\in I$, $E_{c_1}$ has an exclusive transition w.r.t. $I$ or there is some $c_2\in I$ and a transition $j\to k$ in $\text{In}(E_{c_1}) \cap \text{In}(E_{c_1})$ s.t. $V^{jk}_{c_1 c_1} - V^{jk}_{c_1 c_2}$ is monotonically increasing.
	\end{enumerate}
\end{lemma}
\begin{proof}
	The first condition ensures irreducibility of the matrix $\mathbf{V}_I$. The second condition ensures diagonal dominance. The third condition implies that for at least one diagonal entry, strict inequality in the diagonal dominance formulation can be guaranteed. Hence, $\mathbf{V}_I$ is irreducibly diagonally dominant. Such matrices are known to be non-singular.
\end{proof}
The condition \ref{item:strict dominance} is fulfilled if there is a positive probability that a patient enters the set $E_{c_2}$ before the transition $j \to k$ is made. A similar result can be obtained if the condition \ref{item:irreducibility} is not given. This follows from a simple decomposition of the covariance function.
\begin{cor}
	Let \ref{item:irreducibility} of Lemma \ref{lemma:irreducibility} not be fulfilled. Let $I_1, \dots, I_b$ denote a disjunct decomposition of $I$  s.t. \ref{item:irreducibility} of Lemma \ref{lemma:irreducibility} is fulfilled for any $a \in \{1,\dots, b\}$. If conditions \ref{item:diagonal_dominance} and \ref{item:strict dominance} of Lemma \ref{lemma:irreducibility} are fulfilled for at least one of $I_1, \dots, I_b$, then the increments of 
\end{cor}
\begin{proof}
	After rearrangement of the entries, the covariance function $\mathbf{V}_I$ is block diagonal and can be written as
	\begin{equation*}
		\mathbf{V}_I = \text{diag}(\mathbf{V}_{I_1}, \dots, \mathbf{V}_{I_b}).
	\end{equation*}
	The increments of this function are positive definite if the increments of $\mathbf{V}_{I_1}, \dots, \mathbf{V}_{I_b}$ are positive semi-definite and at least one of them is positive definite. For one of $\mathbf{V}_{I_1}, \dots, \mathbf{V}_{I_b}$, Lemma \ref{lemma:irreducibility} can be applied by assumption.
\end{proof}

\subsection*{A.3 Adaptive case study}
While the simulation study in the main manuscript shows how our procedure enables the simulatenous assessment of PFS and OS as co-primary endpoints and the real data example shows how results of the procedure can be displayed and interpreted, we want to focus on interim trial adaptations in this case study. As explained above, it is allowed to perform data-dependent sample size recalculations based on all involved time-to-event endpoints. This is not covered by previous methodology.\\
Here, we illustrate an adaptive testing strategy based on the specific illness-death model from \cite{Fleischer:2009} for patients with non-small cell lung cancer. The dependence between PFS and OS is modelled by a time-homogeneous Markov model, i.e. we have $\gamma^{01} = \gamma^{02} = \gamma^{12} = 1$ in terms of \eqref{eq:weibull_intensities}. The transition rates are given by $\lambda^{01}=0.284$, $\lambda^{02}=0.075$ and $\lambda^{12}=0.128$. It should be noted that the time scale for this model is given in months.\\
For our case study, we assume that this model states the course of disease under the standard of care which shall now be compared to a new experimental treatment in a clinical trial. The calendar time parameters of this fictional trial shall be given as follows: accrual period $a = 24$ with uniform accrual rate, fixed follow-up time $f = 12$, an interim analysis at $t_1 = 18$ and a final analysis at $t_2 = a + f = 36$. The PFS- and OS-rates in the model explained above are given by 63.4\% resp. 41.6\% at $t_1$ and 99.8\% resp. 91.8\% at $t_2$. We want to apply the testing procedure introduced in the main manuscript to detect differences between the distribution of PFS and OS between the two treatment groups. Stagewise $p$-values will be computed as in \eqref{eq:stagewise_p_values_pfsos} and combined via the inverse normal combination function with equal weights. An overall type I error level of $\alpha = 5\%$ shall be achieved by using the sequential O'Brien-Fleming boundaries.\\
As in the simulation study, this fictional trial is planned under proportional hazards assumptions for all three transitions with hazard ratios $\delta^{01}=1/1.5$, $\delta^{02}=1$ and $\delta^{12}=1/1.25$. If 20 patients are be recruited per month, a power of just about 80\% is achieved under this planning assumptions. This amounts to an overall sample size of 480. Obviously, the actual power may differ from this targeted value if at least one of the hazard rations is misspecified. Applying the results obtained here, we want to use the complete information about PFS- and OS-events available at calendar time $t_1$ to reschedule the remaining trial. The adaptation rule at the interim analysis is inspired by the rule applied in the NB2004-HR trial \cite{Berthold:2020}. Transition intensities $\lambda^{01}$, $\lambda^{02}$ and $\lambda^{12}$ will be estimated in both groups based on the available data according to the maximum likelihood procedure given in \cite{Fleischer:2009}. Based on these estimates, the conditional power of the procedure will be calculated under the current design as well as under some alternative designs with adaptations of the duration of the accrual period which will always be followed by the follow-up period of fixed length $f$. Recruitment will be continued for a flexible time $a_{\text{add}}$ which shall be bounded from below and above by $a_{\text{add, min}} \leq a_{\text{add}} \leq a_{\text{add, max}}$. If $\psi(\cdot)$ denotes the conditional power under the new planning alternative in dependence of the additional recruitment peroid, the decision rule on how to proceed with the trial is given as follows:
\begin{equation*}
	a_{\text{add}} = 
	\begin{cases}
		a_{\text{add, min}} & \text{if } \psi(a_{\text{add, min}}) \geq 0.8\\
		\psi^{-1}(0.8) & \text{else if } \psi(a_{\text{add, min}}) < 0.8$ \text{ and } $\psi(a_{\text{add, max}}) \geq 0.8\\
		a_{\text{add, max}} & \text{else if } \psi(a_{\text{add, max}}) \geq 0.5\\
		a_{\text{add, min}} & \text{else.}
	\end{cases}
\end{equation*}
In this study we consider the choices $a_{\text{add, min}} = 3$ and $a_{\text{add, max}} \in \{30, 42\}$ which would amount to a doubling of the originally planned accrual rate or a doubling of the originally planned overall trial duration, respectively.\\
We simulated scenarios with any of the combinations of the following "true" hazard ratios:
\begin{equation*}
	\delta^{01} \in \{1.3^{-1}, 1.4^{-1}, 1.5^{-1}, 1.6^{-1}, 1.7^{-1}\} \text{ and } \delta^{12} \in \{1.15^{-1}, 1.2^{-1}, 1.25^{-1}, 1.3^{-1}, 1.35^{-1}\}.
\end{equation*}
For each combination, 10,000 simulation runs were executed. The empirical rejection rates of the three procedures (group-sequential, adaptive with $a_{\text{add, max}} = 30$ and adaptive with $a_{\text{add, max}} = 42$) are displayed in Table \ref{table:adaptive_vs_gs}.

\begin{table}
	\centering
	\caption{Comparison of empirically attained power of group-sequential and adaptive designs for several deviations from the planning hypothesis (originally planned sample size is set for hazard ratios $\delta^{01}=1.5^{-1}$ and $\delta^{12}=1.25^{-1}$). Upper value in each cell refers to power of group-sequential design, middle value refers to adaptive design with $a_{\text{add, max}} = 30$, lower value refers to adaptive design with $a_{\text{add, max}} = 42$. Numbers in brackets show average duration of accrual period in the respective scenario (in months).}
	
	\begin{tabular}{|c|c||c|c|c|c|c|}
		\cline{3-7}
		\multicolumn{2}{c}{}&\multicolumn{5}{|c|}{$\delta_{01}$}\\
		\cline{3-7}
		\multicolumn{2}{c|}{}&$1/1.3$&$1/1.4$&$1/1.5$&$1/1.6$&$1/1.7$\\
		\hhline{--::=====}
		\multirow{15}{*}{$\delta_{12}$}&\multirow{3}{*}{$1/1.15$}
		\unskip\textcolor{white}{\makebox[0pt]{\smash{\rule[9pt]{14pt}{3pt}}}}
		&0.416 (23.24)&0.607 (22.54)&0.778 (21.63)&0.894 (20.75)&0.957 (19.97)\\
		&&0.451 (25.61)&0.645 (23.66)&0.791 (22.02)&0.888 (20.63)&0.940 (19.65)\\
		&&0.470 (27.31)&0.656 (24.87)&0.799 (22.48)&0.887 (20.79)&0.944 (19.79)\\
		\hhline{~-||-----}
		&\multirow{3}{*}{$1/1.2$}
		&0.458 (23.18)&0.639 (22.47)&0.790 (21.67)&0.899 (20.77)&0.951 (19.98)\\
		&&0.494 (25.45)&0.663 (23.64)&0.794 (21.95)&0.883 (20.63)&0.945 (19.64)\\
		&&0.513 (26.95)&0.683 (24.64)&0.805 (22.59)&0.890 (21.02)&0.943 (19.74)\\
		\hhline{~-||-----}
		&\multirow{3}{*}{$1/1.25$}
		&0.502 (23.13)&0.670 (22.48)&0.806 (21.70)&0.900 (20.80)&0.954 (20.00)\\
		&&0.525 (24.87)&0.695 (23.37)&0.812 (21.79)&0.897 (20.58)&0.942 (19.56)\\
		&&0.555 (26.18)&0.703 (24.30)&0.820 (22.35)&0.892 (20.82)&0.945 (19.76)\\
		\hhline{~-||-----}
		&\multirow{3}{*}{$1/1.3$}
		&0.563 (23.04)&0.709 (22.45)&0.828 (21.59)&0.911 (20.80)&0.960 (20.02)\\
		&&0.585 (24.62)&0.717 (23.06)&0.827 (21.54)&0.901 (20.48)&0.947 (19.55)\\
		&&0.618 (25.75)&0.727 (23.90)&0.829 (22.15)&0.907 (20.58)&0.948 (19.71)\\
		\hhline{~-||-----}
		&\multirow{3}{*}{$1/1.35$}
		&0.615 (22.99)&0.750 (22.34)&0.848 (21.62)&0.926 (20.77)&0.964 (19.98)\\
		&&0.640 (23.93)&0.756 (22.69)&0.849 (21.49)&0.913 (20.35)&0.951 (19.60)\\
		&&0.654 (24.98)&0.762 (23.22)&0.853 (21.82)&0.915 (20.46)&0.955 (19.59)\\
		\hhline{--||-----}
	\end{tabular}
	\label{table:adaptive_vs_gs}
\end{table} 

The breadth of the 95\%-confidence intervals for an underlying value of 0.8 amounts to about 0.016. The adaptive procedures assume the preplanned power if the planning assumptions agree with the truth. The power of these procedures does not increase more than that of the group-seuqnetial design if one of the hazard ratios is underestimated. The average duration of the procedures in these cases is very similar. In contrast, the type II error rate decreases remarkably when using an adaptive design instead of the group-sequential design in case differences between the two groups are overestimated. Especially if $\delta^{01}$ is overestimated, a difference in post-progression survival can help to recover some of the lost power. However, this increases the average duration of the recruitment phase by up to 4 months.\\
Concerning maintenance of the nominal level of the procedure, all three procedures perform accetably (group-sequential procedure: 4.95\%, adaptive procedure with $a_{\text{add, max}} = 30$: 5.06\%, adaptive design with $a_{\text{add, max}} = 42$: 5.24\%; 95\%-confidence interval: $[4.57\%, 5.43\%]$).

\end{document}